\documentclass{amsart}

\usepackage{amsmath,amssymb,amsthm} 
\usepackage{bbm}
\usepackage{stmaryrd}
\usepackage{graphicx}

\renewcommand\i{\infty}
\newcommand\ve{\varepsilon}

\newcommand\T{0\le t\le T}

\newcommand\I{\mathbbm{1}}

\newcommand\p{\mathbb{P}}

\newcommand\E{\mathbb{E}}

\newtheorem{theorem}{Theorem}[section]

\newtheorem{corollary}[theorem]{Corollary}

\newtheorem{definition}[theorem]{Definition}
\newtheorem{lemma}[theorem]{Lemma}
\newtheorem{proposition}[theorem]{Proposition}

\theoremstyle{definition}

\newtheorem{remark}[theorem]{Remark}

\numberwithin{equation}{section}

\begin{document}

\title[Growth-optimal portfolio under transaction costs]{The dual optimizer for the growth-optimal portfolio under transaction costs}
\author{Stefan Gerhold}
\address{Institute for Mathematical Methods in Economics \endgraf Vienna University of Technology \endgraf Wiedner Hauptstrasse 8-10, A-1040 Wien, Austria}
\email{sgerhold@fam.tuwien.ac.at}
\author{Johannes Muhle-Karbe}
\address{Departement f\"ur Mathematik, ETH Z\"urich\endgraf R\"amistrasse 101, CH-8092 Z\"urich, Switzerland}
\email{johannes.muhle-karbe@math.ethz.ch}
\author{Walter Schachermayer}
\address{Fakult\"at f\"ur Mathematik, Universit\"at Wien \endgraf Nordbergstrasse 15, A-1090 Wien, Austria}
\email{walter.schachermayer@univie.ac.at}
\date{\today}

\thanks{We sincerely thank Paolo Guasoni, Jan Kallsen, and Mikl\'os R{\'a}sonyi for helpful discussions, in particular during the early stages of the project.}

\begin{abstract}
We consider the maximization of the long-term growth rate in the Black-Scholes model under proportional transaction costs as in Taksar, Klass and Assaf [{\em Math.\ Oper.\ Res.}~13, 1988]. Similarly as in Kallsen and Muhle-Karbe [\emph{Ann.\ Appl.\ Probab.,} 20, 2010] for optimal consumption over an infinite horizon, we tackle this problem by determining a \emph{shadow price}, which is the solution of the dual problem. It can be calculated explicitly up to determining the root of a deterministic function. This in turn allows to explicitly compute fractional Taylor expansions, both for the no-trade region of the optimal strategy and for the optimal growth rate.
\end{abstract}

\maketitle

\section{Introduction}
Portfolio optimization is a classical example of an infinite-dimensional concave optimization problem. The first ingredient is a probabilistic model of a financial market, e.g., the Black-Scholes model consisting of a \emph{bond}  modelled as
\begin{equation}\label{eq:bond}
S^0_t=\exp(rt)
\end{equation}
and a \emph{stock} modelled as
\begin{equation}\label{eq:stock}
S_t=S_0 \exp\left(\sigma W_t+\left(\mu-\frac{\sigma^2}{2}\right)t\right).
\end{equation}
Here~$W$ is a standard Brownian motion and~$r$, $\mu$ as well as $\sigma, S_0>0$ denote constants. In the sequel, we focus on the Black-Scholes model and assume (without loss of generality for the present purposes) that $S_0=1$, $r=0$ and $\mu>0$.

In order to model the preferences of an economic agent, the second ingredient is a utility function $U: \mathbb{R}_+ \to \mathbb{R} \cup \{-\infty\}$. In the present paper we will deal with the most tractable specification, namely logarithmic utility
\begin{equation*}\label{eq:log}
U(x)=\log(x).
\end{equation*}

The third ingredient is an initial endowment of~$x$ units of bonds, as well as a time horizon $T \in (0,\infty]$.  

There are essentially two versions of the portfolio optimization problem. 

The first version consists of maximizing the \emph{expected utility from consumption}, which is typically formulated for an infinite horizon:
\begin{equation}\label{eq:consumption}
\mathbb{E}\left[\int_0^\infty e^{-\rho t}U(c_t)dt\right] \to \max!
\end{equation}
Here, $\rho>0$ is a discount factor pertaining to the impatience of the investor and $(c_t)_{t \ge 0}$ runs through all positive consumption plans which can be financed by the initial endowment~$x>0$ and subsequent trading in the stock~$S$. In Merton's seminal paper~\cite{merton.69}, it is shown that -- in the Black-Scholes model and for the case of logarithmic or power utility -- there are two constants $\pi,c$, depending on the model parameters, such that the optimal strategy consists of investing a fraction~$\pi$ of the current wealth into the stock and consuming with an intensity which is a fraction~$c$ of the current wealth.

The second version of the portfolio optimization problem is to choose a time horizon~$T$ and to maximize \emph{expected utility from terminal wealth}:
\begin{equation}\label{eq:terminal}
\mathbb{E}\left[U\Big(x+\int_0^T \varphi_t dS_t\Big) \right] \to \max!
\end{equation}
Here we maximize over all predictable processes $\varphi=(\varphi_t)_{t \ge 0}$ describing the number of stocks which the agent holds at time~$t$. We only consider those strategies~$\varphi$ which are \emph{admissible}, i.e.\ lead to a nonnegative wealth process $(x+\int_0^t \varphi_udS_u)_{0 \le t \le T}$. Again, it turns out that -- for the Black-Scholes model~\eqref{eq:stock} and logarithmic or power utility -- the optimal strategy is to keep the proportion~$\pi_t$ between wealth invested in the stock and total wealth constant. In particular, for logarithmic utility, this \emph{Merton rule} reads as
\begin{equation}\label{eq:merton}
\pi_t=\frac{\varphi_t S_t}{\varphi^0_t+\varphi_t S_t}=\frac{\mu}{\sigma^2}.
\end{equation}
Here, $(\varphi^0_t)_{0 \le t \le T}$  and $(\varphi_t)_{0 \le t \le T}$ denote the the holdings in bond and stock, respectively, which are related via the self-financing condition that no funds are added or withdrawn. In fact,~\eqref{eq:merton} holds true much more generally; e.g., for It\^o processes with -- say -- bounded coefficients one just has to replace~$\mu$ and~$\sigma$ with the drift coefficient~$\mu_t$ resp.\ the diffusion coefficient~$\sigma_t$ (cf.\ e.g.~\cite[Example~6.4]{karatzas.al.91}). This particular tractability of the $\log$-utility maximization problem is a fact which we are going to exploit later on.   

\medskip

We now pass to the theme of the present paper, which is portfolio optimization under (small) \emph{transaction costs}. To this end, we now assume that~\eqref{eq:stock} defines the \emph{ask price} of the stock, while the corresponding \emph{bid price} is supposed to be given by $(1-\lambda)S$ for some constant $\lambda \in (0,1)$. This means that one has to pay the higher price~$S_t$ when purchasing the stock at time~$t$, but only receives the lower price $(1-\lambda)S_t$ when selling it.\footnote{This notation, also used in~\cite{taksar.al.88}, turns out to be convenient in the sequel. It is equivalent to the usual setup with the same constant proportional transaction costs for purchases and sales (compare e.g.~\cite{davis.norman.90, janecek.shreve.04, shreve.soner.94}). Indeed, set $\check{S}=\frac{2-\lambda}{2}S$ and $\check{\lambda}=\frac{\lambda}{2-\lambda}$. Then $((1-\lambda)S,S)$ coincides with $((1-\check{\lambda})\check{S},(1+\check{\lambda})\check{S})$. Conversely, any bid-ask process $((1-\check{\lambda})\check{S},(1+\check{\lambda})\check{S})$ with $\check{\lambda} \in (0,1)$ equals $((1-\lambda)S,S)$ for $S=(1+\check{\lambda})\check{S}$ and $\lambda=\frac{2\check{\lambda}}{1+\check{\lambda}}$.} Since transactions of infinite variation lead to immediate bankruptcy, we confine ourselves to the following set of trading strategies.

\begin{definition}
A \emph{trading strategy} is an $\mathbb{R}^2$-valued predictable finite variation process $ (\varphi^0, \varphi)=(\varphi^0_t,\varphi_t)_{t \ge 0}$, where $(\varphi^0_{0-}, \varphi_{0-})=(x,0)$ represents the initial endowment in bonds\footnote{This assumption is made mainly for notational convenience. An extension to general initial endowments is straightforward.} and $\varphi_t^0, \varphi_t$ denote the number of shares held in the bank account and in stock at time~$t$, respectively. 
\end{definition}

To capture the notion of a self-financing strategy, we use the intuition that no funds are added or withdrawn. To this end, we write the second component~$\varphi$ of a strategy $(\varphi^0,\varphi)$ as the difference $\varphi=\varphi^{\uparrow}-\varphi^{\downarrow}$ of two increasing processes~$\varphi^{\uparrow}$ and~$\varphi^{\downarrow}$ which do not grow at the same time. The proceeds of selling stock must be added to the bank account while the expenses from the purchase of stock have to be deducted from the bank account in
any infinitesimal period $(t -dt, t]$, i.e., we require
\begin{equation*}\label{e:selff1}
d\varphi^0_t=(1-\lambda) S_{t}d\varphi^{\downarrow}_t -  S_{t}d\varphi^{\uparrow}_t.
\end{equation*}
Written in integral terms this amounts to the following notion.

\begin{definition}\label{defi:selffinancing}
A trading strategy $(\varphi^0,\varphi)_{t \ge 0}$ is called \emph{self-financing}, if
\begin{equation}\label{e:selff2}
\varphi^0=\varphi^0_{0-}+\int_0^\cdot (1-\lambda) S_{t} d\varphi^{\downarrow}_t - \int_0^\cdot S_{t} d\varphi^{\uparrow}_t,
\end{equation}
where $\varphi=\varphi^{\uparrow}-\varphi^{\downarrow}$ for increasing predictable processes $\varphi^{\uparrow},\varphi^{\downarrow}$ which do not grow at the same time.
\end{definition}

Note that since~$S$ is continuous and~$\varphi$ is of finite variation, integration by parts yields that this definition coincides with the usual notion of self-financing strategies in the absence of transaction costs if we let $\lambda=0$.

The subsequent definition requires the investor to be solvent at all times. For frictionless markets, i.e.\ if $\lambda=0$, this coincides with the usual notion of admissibility.

\begin{definition}\label{def:opti}
A self-financing trading strategy $(\varphi^0,\varphi)_{t \ge 0}$ is called \emph{admissible}, if its \emph{liquidation wealth process} 
$$V_t(\varphi^0,\varphi):=\varphi_t^0+\varphi_t^+(1-\lambda)S_t-\varphi_t^-S_t, \quad t\geq 0, $$
is a.s.\ nonnegative.  
\end{definition}

Utility maximization problems under transaction costs have been studied extensively. In the influential paper~\cite{davis.norman.90}, Davis and Norman identify the solution to the infinite-horizon consumption problem~\eqref{eq:consumption} (compare also~\cite{janecek.shreve.04,shreve.soner.94}). Transaction costs make it unfeasible to keep a fixed proportion of wealth invested into stocks, as this would involve an infinite variation of the trading strategy. Instead, it turns out to be optimal to keep the fraction~$\pi_t$ of wealth in stocks in terms of the ask price~$S_t$ inside some interval. Put differently, the investor refrains from trading until the proportion of wealth in stocks leaves a \emph{no-trading region}. The boundaries of this no-trade region are not known explicitly, but can be determined numerically by solving a free boundary problem.

Liu and Loewenstein~\cite{liu.loewenstein.02} approximate the finite horizon problem by problems with a random horizon, which turn out to be more tractable. Dai and Yi~\cite{dai.yi.09} solve the finite-horizon problem by characterizing the time-dependent boundaries of the no-trade region as the solution to a double-obstacle problem, where the ODE of~\cite{davis.norman.90} is replaced by a suitable PDE. Taksar et al.~\cite{taksar.al.88} consider the long-run limit of the finite horizon problem, i.e.\ the maximization of the portfolio's asymptotic logarithmic growth rate. As in the infinite-horizon consumption problem, this leads to a no transaction region with constant boundaries. However, these boundaries are determined more explicitly as the roots of a deterministic function. Arguing on an informal level, Dumas and Luciano~\cite{dumas.luciano.91} extend this approach to the maximization of the asymptotic power growth rate. To the best of our knowledge, a rigorous proof of this result still seems to be missing in the literature, though.

\medskip
 
{}From now on, we only consider logarithmic utility $U(x)=\log(x)$ and formulate the problem, for given initial endowment $x>0$ in bonds, time horizon $T \in (0,\infty)$  and transaction costs $\lambda \in (0,1)$, in direct analogy to the frictionless case in~\eqref{eq:terminal} above.

 \begin{definition}[$\log$-optimality for horizon~$T$, first version]
 An admissible strategy $(\varphi^0,\varphi)_{0 \le t \le T}$  is called \emph{$\log$-optimal} on $[0,T]$ for the bid-ask process $((1-\lambda)S,S)$, if 
\begin{equation}\label{eq:tacopt1}
\E\left[\log(V_T(\psi^0,\psi))\right] \le \E\left[\log(V_T(\varphi^0,\varphi))\right],
\end{equation}
for all competing admissible strategies $(\psi^0_t,\psi_t)_{0 \le t \le T}$. 
\end{definition}

It turns out that this problem is rather untractable. To see this, consider the special and particularly simple case $\mu=\sigma^2$. In the frictionless case, Merton's rule~\eqref{eq:merton} tells us what the optimal strategy is: At time~$0$, convert the entire initial holdings into stock, i.e., pass from $(\varphi^0_{0-},\varphi_{0-})=(x,0)$ to $(\varphi^0_0,\varphi_0)=(0,x/S_0)=(0,x)$. Then keep all the money in the stock, i.e., $(\varphi^0_t,\varphi_t)=(0,x)$, for the entire period $[0,T]$. At the terminal date~$T$, this provides a logarithmic utility of $\log(xS_T)$, after converting the~$x$ stocks into~$xS_T$ bonds (without paying transaction costs).

Let us now pass to the setting with transaction costs $\lambda>0$. If $\lambda \ll T$, it is, from an economic point of view, rather obvious what constitutes a ``good'' strategy for the optimization problem~\eqref{eq:tacopt1}: Again convert the initial holdings of~$x$ bonds at time~$0$ into stocks and simply hold these stocks until time~$T$ without doing any dynamic trading. Converting the stocks back into bonds at time~$T$, this leads to a logarithmic utility of $\log((1-\lambda)xS_T)=\log(xS_T)+\log(1-\lambda)$. Put differently, the difference to the frictionless case is only the fact that at terminal date~$T$ you \emph{once} have to pay the transaction costs $\lambda>0$.

Now consider the case  $0<T \ll \lambda$. In this situation, the above strategy does not appear to be a ``good'' approach to problem~\eqref{eq:tacopt1} any more. The possible gains of the stock during the (short) interval $[0,T]$ are outweighed by the (larger) transaction costs $\lambda$. Instead, it now seems to be much more appealing to simply keep your position of~$x$ bonds during the interval $[0,T]$ and not to invest into the stock at all. 

These considerations are of course silly from an economic point of view, where only the case $0 \le \lambda \ll T$ is of interest. The economically relevant issue is how the \emph{dynamic trading} during the interval $(0,T)$ is affected when we pass from the frictionless case $\lambda=0$ to the case $\lambda>0$. Paying the transaction costs only \emph{once} at time $t=T$ (resp.\ twice if we also model the transaction costs for the purchase at time $t=0$) can be discarded from an economic point of view, as opposed to the ``many'' trades necessary to manage the portfolio during $(0,T)$ if $\mu \neq \sigma^2$. This economic intuition will be made mathematically precise in Corollary~\ref{cor:notradeasymp} and Proposition~\ref{prop:growthasymp} below, where the leading terms of the relevant Taylor expansions in~$\lambda$ are of the order $\lambda^{1/3}$ and $\lambda^{2/3}$, respectively. The effect of paying transaction costs once, however, is only of order~$\lambda$ (compare Corollary~\ref{cor:tac} below). 

Mathematically speaking, a consequence of the above formulation~\eqref{eq:tacopt1} is the loss of time consistency, which we illustrated above for the special case $\mu=\sigma^2$. For $U(x)=\log(x)$, it follows from Merton's rule \eqref{eq:merton} that the optimal strategy in the problem~\eqref{eq:terminal} without transaction costs does not depend on the time horizon~$T$, i.e., is optimal for all $T>0$. In the presence of transaction costs, this desirable concatenation property does not hold true any more for Problem~\eqref{eq:tacopt1} as we have just seen. There is a straightforward way to remedy  this  nuisance, namely passing to the limit $T \to \infty$. This has been done by Taksar, Klass and Assaf~\cite{taksar.al.88} and in much of the subsequent literature.

\begin{definition}
An admissible strategy $(\varphi^0,\varphi)$ is called \emph{growth-optimal}, if 
$$\limsup_{T \to \infty} \frac{1}{T} \E\left[\log(V_T(\psi^0,\psi))\right] \leq \lim_{T \to \infty} \frac{1}{T} \E\left[\log(V_T(\varphi^0,\varphi))\right] ,$$
for all competing admissible strategies $(\psi_t^0,\psi_t)_{t \ge 0}$.
\end{definition}

Note that the optimal growth rate does not depend on the initial endowment~$x$.  Moreover, the above notion does not yield a unique optimizer. As the notion of growth optimality only pertains to a limiting value, suboptimal behaviour on any compact subinterval of $[0,\infty)$ does not matter as long as one eventually behaves optimally. While the notion of growth optimality allows to get rid of the nuisance of terminal liquidation costs, the non-uniqueness of an optimizer has serious drawbacks. For example, much of the beauty of duality theory, which works nicely when the primal and dual optimizers are unique, is lost.

In order to motivate our final remedy to the ``nuisance problem'' (cf.\ Definition~\ref{def:modified} below), we introduce, as in~\cite{kallsen.muhlekarbe.10}, the concept of a shadow price which will lead us to the notion of a dual optimizer .

\begin{definition}\label{def:shadow}
A \emph{shadow price} for the bid-ask process $((1-\lambda)S,S)$ is a continuous semimartingale $\tilde{S}=(\tilde{S}_t)_{t \ge 0}$ with $\tilde{S}_0=S_0$ and taking values in $[(1-\lambda)S,S]$, such that the $\log$-optimal portfolio $(\varphi_t^0,\varphi_t)_{t \ge 0}$ \emph{for the frictionless market with price process~$\tilde{S}$} exists, is of finite variation and the number of stocks~$\varphi$ only increases (resp.\ decreases) on the set $\{\tilde{S}_t=S_t\} \subset \Omega \times \mathbb{R}_+$ (resp.\ $\{\tilde{S}_t=(1-\lambda)S_t\})$. Put differently,~$\varphi$ is the difference of the increasing predictable processes  $\varphi^{\uparrow}=\int_0^\cdot \I_{\{\tilde{S}_t=S_t\}} d\varphi_t$  and $\varphi^{\downarrow}=-\int_0^\cdot \I_{\{\tilde{S}_t=(1-\lambda)S_t\}}d\varphi_t$. 
\end{definition}

We now pass to the decisive trick to modify the \emph{finite-horizon} problem. Given a shadow price~$\tilde{S}$, formulate the optimization problem such that we only allow for trading under transaction costs $\lambda>0$ during the interval $[0,T)$, but at time~$T$ we \emph{make an exception}. At the terminal time~$T$, we allow to liquidate our position in stocks at the shadow price~$\tilde{S}_T$ rather than at the (potentially lower) bid price $(1-\lambda)S_T$. Here is the mathematical formulation:

\begin{definition}[$\log$-optimality for horizon~$T$, modified version]\label{def:modified}
Given a shadow price $\tilde{S}=(\tilde{S}_t)_{t \ge 0}$ and a finite time horizon~$T$, we call an admissible (in the sense of Definition~\ref{def:opti}) trading strategy $(\varphi^0,\varphi)=(\varphi^0_t,\varphi_t)_{0 \le t \le T}$ \emph{$\log$-optimal for the modified problem} if  
$$\E\left[\log(\tilde{V}_T(\psi^0,\psi))\right] \le \E\left[\log(\tilde{V}_T(\varphi^0,\varphi))\right]$$
for every competing admissible strategy $(\psi^0,\psi)=(\psi^0_t,\psi_t)_{0 \le t \le T}$, where
$$\tilde{V}_t(\varphi^0,\varphi):=\varphi^0_t+\varphi_t \tilde{S}_t, \quad t\ge 0,$$
denotes the \emph{wealth process for liquidation in terms of~$\tilde{S}$}.
\end{definition}

Of course, the above definition is ``cheating'' by using the shadow price process~$\tilde{S}$ -- which is part of the solution -- in order to \emph{define} the optimization problem. But  this trick pays handsome dividends: Suppose that the $\log$-optimizer $(\varphi^0,\varphi)=(\varphi^0_t,\varphi_t)_{t \ge 0}$ for the frictionless market~$\tilde{S}$ is \emph{admissible} for the bid-ask process $((1-\lambda)S,S)$, i.e., is of finite variation and has a positive liquidation value even in terms of the lower bid price (this will be the case in the present context). Then this process $(\varphi^0_t,\varphi_t)_{0 \le t \le T}$ is the optimizer for the modified optimization problem from Definition~\ref{def:modified}:

\begin{proposition}\label{prop:shadow}
Let~$\tilde{S}$ be a shadow price for the bid-ask process $((1-\lambda)S,S)$ with associated $\log$-optimal portfolio $(\varphi^0,\varphi)$. If $V(\varphi^0,\varphi) \ge 0$, this portfolio is also $\log$-optimal for the modified problem under transaction costs from Definition \ref{def:modified}. 
\end{proposition}

\begin{proof}
Since~$\varphi$ only increases (resp.\ decreases) on $\{\tilde{S}_t=S_t\}$ (resp.\ $\{\tilde{S}_t=(1-\lambda)S_t\}$), it follows {}from the definition that the portfolio $(\varphi^0,\varphi)$ is self-financing for the bid-ask process $((1-\lambda)S,S)$. Hence it is admissible in the sense of Definition~\ref{def:opti} if $V(\varphi^0,\varphi) \ge 0$. Now let $(\psi^0,\psi)$ be any admissible policy for $((1-\lambda)S, S)$ and set $\tilde{\psi}^0_t:= \psi^0_0 -\int_0^t \tilde{S}_s d\psi_s$. Then $\tilde{\psi}^0 \geq \psi^0$  and $(\tilde{\psi}^0,\psi)$ is an admissible portfolio for~$\tilde{S}$, since $(1-\lambda)S \leq \tilde{S} \leq S$. Together with the $\log$-optimality of $(\varphi^0,\varphi)$ for~$\tilde{S}$, this implies
\begin{align*}
\E\left[\log(\tilde{V}_T(\psi^0,\psi))\right] &\le \E\left[\log(\tilde{V}_T(\tilde{\psi}^0,\psi))\right] \le \E\left[\log(\tilde{V}_T(\tilde{\varphi}^0,\varphi))\right],
\end{align*}
which proves the assertion. 
\end{proof}

As a corollary, we obtain that the difference between the optimal values for the modified and the original problem is bounded by $\log(1-\lambda)$ and therefore of order~$O(\lambda)$ as the transaction costs~$\lambda$ tend to zero. In particular, this difference vanishes if one considers the infinite-horizon problem studied by~\cite{taksar.al.88}. 

\begin{corollary}\label{cor:tac}
Let~$\tilde{S}$ be a shadow price for the bid-ask process $((1-\lambda)S,S)$ with $\log$-optimal portfolio $(\varphi^0,\varphi)$ satisfying $\varphi^0,\varphi \ge 0$. Then 
$$ \E\left[\log(V_T(\varphi^0,\varphi))\right]  \ge \sup_{(\psi^0,\psi)} \E\left[\log(V_T(\psi^0,\psi))\right]+\log(1-\lambda), $$
where the supremum is taken over all $(\psi^0,\psi)$, which are admissible for the bid-ask process $((1-\lambda)S,S)$. Moreover, $(\varphi^0,\varphi)$ is growth-optimal for $((1-\lambda)S,S)$.
\end{corollary}

\begin{proof}
Since  $(1-\lambda)S \leq \tilde{S} \leq S$, we have
\begin{equation}\label{eq:obvious}
V(\psi^0,\psi) \le \tilde{V}(\psi^0,\psi) 
\end{equation}
for any admissible $(\psi^0,\psi)$ and it follows from $\varphi^0,\varphi \ge 0$ that
\begin{equation}\label{eq:up}
V(\varphi^0,\varphi) \ge (1-\lambda)\tilde{V}(\varphi^0,\varphi).
\end{equation}
Combining~\eqref{eq:up}, Proposition~\ref{prop:shadow} and~\eqref{eq:obvious} then yields
\begin{align*}
\E\left[\log(V_T(\varphi^0,\varphi))\right]  &\ge \E\left[\log(\tilde{V}_T(\varphi^0,\varphi)\right]+\log(1-\lambda)\\
 &\ge \E\left[\log(\tilde{V}_T(\psi^0,\psi)\right]+\log(1-\lambda)\\
&\ge \E\left[\log(V_T(\psi^0,\psi)\right]+\log(1-\lambda),
\end{align*}
for all $(\psi^0,\psi)$ admissible for $((1-\lambda)S,S)$, which proves the first part of the assertion. It also implies
\begin{align*}
\limsup_{T \to \infty} \frac{1}{T} \E\left[\log(V_T(\varphi^0,\varphi))\right] &\ge \limsup_{T \to \infty} \frac{1}{T}\left( \E\left[\log(V_T(\psi^0,\psi))\right]+\log(1-\lambda)\right)\\
&= \limsup_{T \to \infty} \frac{1}{T}\E\left[\log(V_T(\psi^0,\psi))\right],
\end{align*}
for any admissible $(\psi^0,\psi)$, which completes the proof. 
\end{proof}

We formulated the corollary only for positive holdings $\varphi^0,\varphi \ge 0$ in bonds and stocks. In the present context, this will only be satisfied if $0 \leq \mu \leq \sigma^2$. To cover also the case $\mu>\sigma^2$, we show in Lemma~\ref{lem:lambdasmall} below that the assertion of Corollary~\ref{cor:tac} remains true more generally in the present setup, provided that the transaction costs~$\lambda$ are \emph{sufficiently small}.

Finally, let us point out that -- due to Definition~\ref{def:modified} -- much of the well-established duality theory for frictionless markets (cf.\ e.g.~\cite{he.pearson.91,karatzas.al.87,pliska.86}) carries over to the modified problem. Let~$\mathbb{Q}_T$ denote the unique equivalent martingale measure for the process $(\tilde{S}_t)_{0 \le t \le T}$. Then the pair $((\tilde{S}_t)_{0 \le t \le T},\mathbb{Q}_T)$, which corresponds to a \emph{consistent price system} in the notation of~\cite{guasoni.al.08}, is the dual optimizer for the modified problem from Definition~\ref{def:modified} (compare~\cite{cvitanic.karatzas.96}). Recalling that the conjugate function to $U(x)=\log(x)$ is $U_c(y)=-\log(y)-1$, we obtain the equality of the primal and dual values 
$$\mathbb{E}\left[\log(\tilde{V}_T(\varphi^0,\varphi))\right]=\mathbb{E}\left[U_c(y\tfrac{d\mathbb{Q}_T}{d\mathbb{P}})\right]+1=\mathbb{E}\left[-\log\left(y\tfrac{d\mathbb{Q}_T}{d\mathbb{P}}\right)\right],$$
where the relation between the Lagrange multiplier $y>0$ and the initial endowment $x>0$ is given by $y=\log'(x)=1/x$. We then also have the first-order condition
$$ \tilde{V}_T(\varphi^0,\varphi)=-U_{c}'(y\tfrac{d\mathbb{Q}_T}{d\mathbb{P}})=x \tfrac{d\mathbb{P}}{d\mathbb{Q}_T}$$
as well as several other identities of the duality theory, see e.g.~\cite{he.pearson.91,karatzas.al.87,pliska.86,kramkov.schachermayer.99}. In other words, the little trick of ``allowing liquidation in terms of a shadow price~$\tilde{S}_T$ at terminal time~$T$'' allows us to use the full strength of the duality theory developed in the frictionless case. 

\medskip 

The main contribution of the present article is that we are able to \emph{explicitly determine} a shadow price process~$\tilde{S}$ for the bid-ask process $((1-\lambda)S,S)$ in Theorem~\ref{thm:shadow}. Roughly speaking, the process~$\tilde{S}$ oscillates between the ask price~$S$ and the bid price  $(1-\lambda)S$, leading to buying (resp.� selling) of the stock when $\tilde{S}_t=S_t$ (resp.\ $\tilde{S}_t=(1-\lambda)S_t$). The predictable sets $\{\tilde{S}_t=S_t\}$ and $\{\tilde{S}_t=(1-\lambda)S_t\}$ when one buys (resp.\ sells) the stock are of ``local time  type''. Remarkably, our shadow price process nevertheless is an It\^o process, whence it ``does not move'' on the sets $\{\tilde{S}_t=S_t\}$ and $\{\tilde{S}_t=(1-\lambda)S_t\}$. The reason is that there is a kind of ``smooth pasting'' when the process~$\tilde{S}$ touches~$S$ resp.\ $(1-\lambda)S$. When this happens, the processes~$\tilde{S}$ and~$S$ (resp.~$\tilde{S}$ and $(1-\lambda)S$) are aligned of first order, see Section~\ref{sec:2} for more details. This parallels the results of~\cite{kallsen.muhlekarbe.10}. These authors determine a shadow price for the infinite-horizon consumption problem. Their characterization, however, involves an SDE with instantaneous reflection, whose coefficients have to be determined from the solution to a free boundary problem. 

Here, on the other hand, the relation between the shadow price~$\tilde{S}$ and the ask price~$S$ (as well as its running minimum and maximum) is established via a deterministic function~$g$, which is the solution of an ODE and known in \emph{closed form} up to determining the root of a deterministic function.  This ODE is derived heuristically from an economic argument in Section~\ref{sec:3}, namely by applying Merton's rule to the process~$\tilde{S}$. Subsequently, we show in Section~\ref{sec:4} that these heuristic considerations indeed lead to well-defined solutions. With our candidate shadow price process~$\tilde{S}$ at hand, Merton's rule quickly leads to the corresponding $\log$-optimal portfolio in Section~\ref{sec:6}. This in turn allows us  to verify that~$\tilde{S}$ is indeed a shadow price. Finally, in Section~\ref{sec:7}, we expound on the explicit nature of our previous considerations. More specifically, we derive fractional Taylor expansions in powers of~$\lambda^{1/3}$ for the relevant quantities, namely the width of the no-trade region and the asymptotic growth rate. The coefficients of these power series, which are rational functions of $(\mu/\sigma^2)^{1/3}$
and $(1-\mu/\sigma^2)^{1/3}$,  can \emph{all} be algorithmically computed. For the related infinite-horizon consumption problem, the leading terms were determined and the second-order terms were conjectured in~\cite{janecek.shreve.04} (compare also~\cite{barles.soner.98, rogers.04, shreve.soner.94, whalley.wilmott.97} for related asymptotic results).

Of course, the very special setting of the paper can be generalized in several directions. One may ask whether similar results can be obtained for more general diffusion processes or, even more generally, for stochastic processes which allow for $\epsilon$-consistent price sytems such as geometric fractional Brownian motion. Another natural extension of the present results is the consideration of power utility and/or consumption. This is a theme for future research (compare \cite{gerhold.al.10b}).

\section{Reflection without local time via smooth pasting}\label{sec:2}

In this section, we show how to construct a process~$\tilde{S}$ that remains within the upper and lower boundaries of the bid-ask spread $[(1-\lambda)S,S]$,  yet does not incorporate local time, i.e., is an It\^o process (see~\eqref{G5a} below).

To this end, suppose that there is a real number $\bar{s} >1$ and a $C^2$-function
\begin{equation}\label{G2b}
g: [1,\bar{s}]\to [1,(1-\lambda)\bar{s}]
\end{equation}
such that $g'(s) >0,$ for $1\le s\le \bar{s},$ and~$g$ satisfies the \emph{smooth pasting condition} with the line $y=x$ at the point $(1,1)$, i.e.,
\begin{equation}\label{G2}
g(1) =g'(1) =1,
\end{equation}
and with the line $y=(1-\lambda)x$ at the point $(\bar{s}, (1-\lambda)\bar{s}),$ i.e.,
\begin{equation}\label{G2a}
g(\bar{s}) =(1-\lambda)\bar{s} \qquad \mbox{and} \qquad g'(\bar{s}) =1-\lambda.
\end{equation}
These conditions are illustrated in Figure~\ref{fig:g} and motivated in Remark \ref{rem:sp} below.
\begin{figure}[htbp]
\centering
\includegraphics[width=0.8\textwidth]{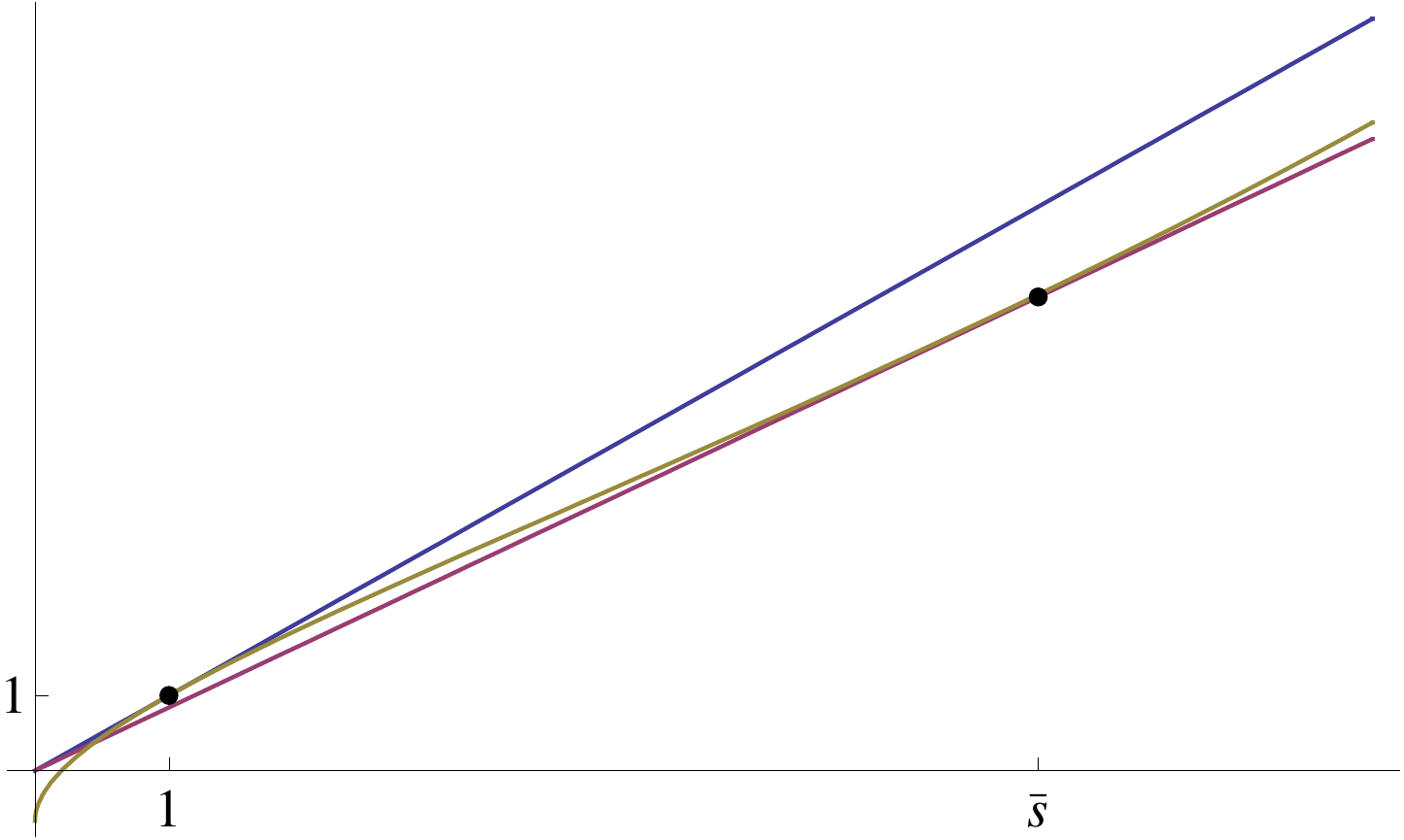}
\caption{Smooth pasting conditions for the function~$g$.}\label{fig:g}
\end{figure}

Now define sequences of stopping times $(\varrho_n)^\i_{n=0}, (\sigma_n)^\i_{n=1}$ and processes $(m_t)_{t \ge 0}$ and $(M_t)_{t \ge 0}$ as follows: let
$\varrho_0=0$ and~$m$ the running minimum process of~$S$, i.e.,
\begin{align*}
m_t=\inf\limits_{\varrho_0 \le u\le t}  S_u, \qquad 0\le t\le\sigma_1,
\end{align*}
where the stopping time~$\sigma_1$ is defined as 
$$\sigma_1 =\inf \{t\geq \varrho_0 : \tfrac{S_t}{m_t} \geq \bar{s}\}.$$

Next define~$M$ as the running maximum process of~$S$ after time~$\sigma_1$, i.e.,
\begin{align*}
M_t = \sup\limits_{\sigma_1 \le u\le t} S_u, \qquad \sigma_1 \le t\le\varrho_1,
\end{align*}
where the stopping time~$\varrho_1$ is defined as
$$\varrho_1 =\inf \{t\geq \sigma_1 : \tfrac{S_t}{M_t} \le \tfrac{1}{\bar{s}}\}.$$

For $t\geq \varrho_1$, we again define
\begin{align*}
m_t=\inf\limits_{\varrho_1\le u\le t} S_u, \qquad \varrho_1 \le t\le \sigma_2,
\end{align*}
where
$$\sigma_2 =\inf \{t\geq \varrho_1 : \tfrac{S_t}{m_t} \geq \bar{s}\},$$
and, for $t\geq \sigma_2$, we define
\begin{align*}
M_t= \sup\limits_{\sigma_2 \le u\le t} S_u, \qquad \sigma_2 \le t\le \varrho_2,
\end{align*}
where
$$\varrho_2 =\inf \{ t\geq \sigma_2 : \tfrac{S_t}{M_t} \le \tfrac{1}{\bar{s}}\}.$$

Continuing in an obvious way we obtain series  $(\varrho_n)^\i_{n=0}$ and $(\sigma_n)^\i_{n=1}$ of a.s.\ finite stopping times $\varrho_n$ and $\sigma_n$, increasing a.s.\ to infinity,
such that~$m$ (resp.~$M$) are the relative running minima (resp.\ maxima) of $S$ defined on the stochastic intervals $(\llbracket \varrho_{n-1}, \sigma_n \rrbracket)^\i_{n=1}$ (resp.\ $(\llbracket\sigma_n, \varrho_n \rrbracket)^\i_{n=1}$ ). Note that
$$\bar{s} m_{\varrho_n} = M_{\varrho_n} = \bar{s} S_{\varrho_n}, \qquad\qquad \mbox{for}\ n \in \mathbb{N},$$
and
$$\bar{s} m_{\sigma_n} =M_{\sigma_n} = S_{\sigma_n}, \qquad \mbox{for} \ n \in \mathbb{N}.$$

We may therefore continuously extend the processes~$m$ and~$M$ to~$\mathbb{R}_+$ by letting
\begin{align*}
M_t :=\bar{s} m_t, \qquad &\mbox{for} \ t\in \bigcup\limits^\i_{n=0} \llbracket\varrho_n, \sigma_{n+1}\rrbracket, \\
m_t := \tfrac{M_t}{\bar{s}}, \qquad &\mbox{for} \ t\in \bigcup\limits^\i_{n=1} \llbracket \sigma_n, \varrho_n \rrbracket.
\end{align*}

For $t \ge 0$, we then have $\bar{s} m_t =M_t$ as well as $m_t \le S_t \le M_t$, and hence
$$ m_t \le S_t \le \bar{s}m_t, \qquad \mbox{for} \ t \ge 0.$$ 
By construction, the processes~$m$ and~$M$ are of finite variation and only decrease (resp.\ increase) on the predictable set $\{m_t=S_t\}$ (resp.\ $\{M_t=S_t\}=\{m_t=S_t/\bar{s}\})$.

\medskip

We can now state and prove the main result of this section.

\begin{proposition}\label{proG5}
Under the above assumptions define the continuous process
\begin{equation}\label{G5}
\tilde{S}_t=m_t g(\tfrac{S_t}{m_t}), \qquad t \ge 0.
\end{equation}

Then~$\tilde{S}$ is an It\^o process starting at $\tilde{S}_0=S_0=1$ and satisfying the stochastic differential equation
\begin{equation}\label{G5a}
d\tilde{S}_t=g'\left(\tfrac{S_t}{m_t}\right) dS_t +\tfrac{1}{2m_t} g'' \left(\tfrac{S_t}{m_t}\right) d \langle S,S \rangle_t.
\end{equation}
Moreover,~$\tilde{S}$ takes values in the bid-ask spread $[(1-\lambda)S,S]$.
\end{proposition}

\begin{remark}
We have formulated the proposition only for the Black-Scholes model~\eqref{eq:stock}. But -- unlike the considerations in the following sections -- it has little to do with this particular process and can also be formulated for general It\^o processes satisfying some regularity conditions.
\end{remark}

\begin{remark}\label{rem:sp}
Formula~\eqref{G5a} is obtained by applying It\^{o}'s formula to~\eqref{G5}, pretending that the process $(m_t)_{t \ge 0}$ were constant.
The idea behind this approach is that on the complement of the ``singular'' set $\{S_t=m_t\} \cup \{S_t=M_t\} \subseteq \Omega \times \mathbb{R}_+$
the process $(m_t)_{t \ge 0}$ indeed ``does not move'' (the statement making sense, at least, on an intuitive level). On the set $\{S_t=m_t\} \cup \{S_t=M_t\}$,
where the process $(m_t)_{t \ge 0}$ ``does move'', the smooth pasting conditions~\eqref{G2} and~\eqref{G2a} will make sure that the SDE~\eqref{G5a} is not
violated either, i.e., the process~$\tilde{S}$ ``does not move'' on this singular set. This intuitive reasoning will be made precise in the subsequent proof of Proposition~\ref{proG5}.
\end{remark}

\begin{proof}[of Proposition~\ref{proG5}]
We first show that the process~$\tilde{S}$ defined in~\eqref{G5} satisfies the SDE~\eqref{G5a} on the stochastic interval $\llbracket 0,\sigma_1 \wedge T\rrbracket$, where $T>0$ is arbitrary.

Fix $0<\ve <\ve_0,$ where $\ve_0=\bar{s}-1,$ and define inductively the stopping times $(\tau_k)^\i_{k=0}$ and $(\eta_k)^\i_{k=1}$ by letting
$\tau_0=0$ and, for $k\geq 1$, 
\begin{align*}
\eta_k &=\inf \{t: \tau_{k-1} <t\le \sigma_1,\ \tfrac{S_t}{m_t} \geq 1+\ve\} \wedge T, \\
\tau_k &=\inf \{t:\eta_k <t\le \sigma_1,\ \tfrac{S_t}{m_t} \le 1+\tfrac{\ve}{2}\} \wedge T.
\end{align*}

Clearly, the sequences $(\tau_k)^\i_{k=0}$ and $(\eta_k)^\i_{k=1}$ increase a.s.~to $\sigma_1\wedge T$.

We partition the stochastic interval $\rrbracket 0,\sigma_1\wedge T\rrbracket$ into $L^\ve \cup R^\ve$ (the letters reminding of
``local time'' and ``regular set''), where
$$L^\ve =\bigcup\limits^\i_{k=1} \rrbracket \tau_{k-1}, \eta_k\rrbracket, \qquad R^\ve =\bigcup\limits^\i_{k=1}\rrbracket \eta_k,\tau_k \rrbracket.$$

As~$R^\ve$ is a predictable set we may form the stochastic integral $\int_0^\cdot  \I_{R^\ve}(u) dm_u$. Arguing on each of the intervals
$\rrbracket \eta_k ,\tau_k\rrbracket$, we obtain
\begin{equation}\label{G8a}
\int_0^t \I_{R^\ve}(u) dm_u =\sum\limits^\i_{k=1} \int^t_0 \I_{\rrbracket\eta_k, \tau_k\rrbracket} (u)~ dm_u =0, \quad  \mbox{for} \ 0\le t\le\sigma_1 \wedge T.
\end{equation}

This is a mathematically precise formula corresponding to the intuition that~$m$ ``does not move'' on~$R^\ve.$ Arguing once more on the intervals
$\rrbracket\eta_k,\tau_k\rrbracket,$ It\^o's formula and~\eqref{G8a} imply that 
\begin{equation}\label{G8}
\int_0^t \I_{R^\ve}(u) d\tilde{S}_u = \int^t_0 \I_{R^\ve} (u) \left[g' \left(\tfrac{S_u}{m_u}\right) dS_u +\tfrac{1}{2m_u} \ g'' \left(\tfrac{S_u}{m_u}\right) 
d \langle S,S \rangle_u\right].
\end{equation}
In other words, the SDE~\eqref{G5a} holds true, when localized to the set~$R^\ve.$

We now show that the process $\int_0^\cdot \I_{L^\ve}(u) d\tilde{S}_u$ tends to zero, as $\ve \to 0.$ More precisely, we shall show that 
\begin{equation}\label{G9}
\lim\limits_{\ve \to 0} \sup\limits_{0\le t\le\sigma_1} \left| \int_0^t  \I_{L^\ve}(u) (d\tilde{S}_u-dS_u)\right| =0,
\end{equation}
where the limit is taken with respect to convergence in probability. This will finish the proof of~\eqref{G5a} on $\llbracket 0, \sigma_1 \wedge T\rrbracket,$ 
as~\eqref{G9} implies that, for $0\le t\le\sigma_1 \wedge T$,
\begin{align*}
\tilde{S}_t&=1+\lim\limits_{\ve\to 0} \int_0^t \I_{R^\ve}(u) d\tilde{S}_u \\
&=1+\lim\limits_{\ve\to 0} \int^t_0 \I_{R^\ve} (u) \left[g'\left(\tfrac{S_u}{m_u}\right) dS_u +\tfrac{1}{2m_u} ~ g''\left(\tfrac{S_u}{m_u}\right) d\langle S,S \rangle_u \right] \\
&=1+\int^t_0 \left[g' \left(\tfrac{S_u}{m_u}\right) dS_u+\tfrac{1}{2m_u} ~ g''\left(\tfrac{S_u}{m_u}\right) d \langle S,S \rangle_u\right].
\end{align*}

Here the first equality follows from~\eqref{G9} and the fact that $\lim_{\ve\to 0}\mathrm{Leb}\ \otimes\ \p(L^\ve)=0$, with $\mathrm{Leb}$ denoting Lebesgue
measure on $[0,T]$, which gives
$$\lim\limits_{\ve\to 0} \sup\limits_{0\le t\le \sigma_1 \wedge T} \left|\int_0^t \I_{L^\ve}(u)dS_u\right|=\lim\limits_{\ve\to 0} \sup\limits_{0\le t\le \sigma_1 \wedge T} \left| \int_0^t \I_{L^\ve}(u)d\tilde{S}_u\right| =0$$
in probability. The second equality is just~\eqref{G8}, and the third one again follows from $\lim_{\ve\to 0}\mathrm{Leb} \otimes\p (L^\ve)=0$ and the fact
that the drift and diffusion coefficients appearing in the above integral are locally bounded. 

To show~\eqref{G9}, fix $\omega\in\Omega$ and $k\geq 1$ such that $\eta_{k+1}(\omega) < \sigma_1(\omega)\wedge T.$
By the definition of~$\tilde{S}$ and~$\tau_k$ as well as a second order Taylor expansion of~$g$ around~$1$ utilizing $g(1)=g'(1)=1$, we obtain 
\begin{align*}
| S_{\tau_k} (\omega) - \tilde{S}_{\tau_k}(\omega)| &= \left|S_{\tau_k}(\omega) -m_{\tau_k}(\omega) g\left(\tfrac{S_{\tau_k}(\omega)}{m_{\tau_k}(\omega)}\right)\right| \\
&=\left|S_{\tau_k}(\omega) \left(1-\tfrac{1}{1+\tfrac{\ve}{2}} \ g (1+\tfrac{\ve}{2})\right)\right| \\
&\le C(\omega)\ve^2,
\end{align*}
where $C(\omega):=\max_{0 \le t \leq T} S_t(\omega) \times \max_{1\le s\le\bar{s}} |g''(s)|$ does not depend on~$\ve$ and~$k$. Likewise,
$|S_{\eta_k}(\omega) - \tilde{S}_{\eta_k}(\omega)| \le C(\omega)\ve^2,$ and, in fact
$$|S_t (\omega)- \tilde{S}_t (\omega)|\le C(\omega)\ve^2,\qquad \mbox{for}  \ \tau_k (\omega)\le t\le\eta_k (\omega),$$
for fixed~$k$. Denote by~$N^\ve$ the random variable
$$N^\ve =\sup \{k\in \mathbb{N}: \ \tau_k <\sigma_1\wedge T\}.$$

Then
\begin{equation}\label{G12}
\sup\limits_{0\le t\le \sigma_1} \left|\int_0^t \I_{L^\ve}(u) (d\tilde{S}_u -dS_u)\right| \le (N^\ve +1)C\ve^2.
\end{equation}

By It\^o's formula, we have
$$d\left(\frac{S_t}{m_t}\right)= \mu \frac{S_t}{m_t}dt+\frac{S_t}{m_t}d\log(m_t)+\sigma\frac{S_t}{m_t}dW_t.$$
Since $|S/m|$ is bounded by~$\bar{s}$, the third term on the right-hand side is a square-integrable martingale, and the first one is of integrable variation. Moreover, the variation of the second term is bounded by~$2\bar{s}$ times the variation of $\sup_{0 \leq t \leq T} \log(S_T)$, which is integrable as well.  As $S_{\tau_k}/m_{\tau_k}-S_{\eta_k}/m_{\eta_k} = -\tfrac{\ve}{2}$, if $\eta_k <\sigma_1 \wedge T,$ one can therefore apply a version of Doob's upcrossing inequality for semimartingales (cf.~\cite{barlow.83}) to conclude that $\lim_{\ve\to 0} \ve^{3/2} N^\ve=0$ in~$L^1$ and hence in probability. Thus~\eqref{G12} implies~\eqref{G9}
which in turn shows~\eqref{G5a}, for $0\le t\le \sigma_1 \wedge T$.

\medskip

Repeating the above argument by considering the function~$g$ in an $\ve$-neighborhood of~$\bar{s}$ rather than~$1$ and using that $g(\bar{s})=(1-\lambda)\bar{s}$ and $g'(\bar{s})=1-\lambda$, we obtain
$$\lim\limits_{\ve\to 0} \sup\limits_{\sigma_1 \wedge T\le t\le \varrho_1\wedge T} \left|\int_0^t \I_{L^\ve}(u) (d\tilde{S}_u -dS_u)\right|=0,$$
which implies the validity of~\eqref{G5a} for $\sigma_1 \wedge T \le t\le\varrho_1 \wedge T.$

Continuing in an obvious way we obtain~\eqref{G5a} on $\bigcup^\i_{k=1} (\rrbracket \varrho_{k-1}, \sigma_k\rrbracket \cup \rrbracket \sigma_k , \varrho_k
\rrbracket ) \cap [0,T]=[0,T]$. Since~$T$ was arbitrary, this completes the proof.
\end{proof}

\begin{remark}\label{rem:symmetric}
We have made the assumption $\bar{s} >1$ in~\eqref{G2b} above. There also is a symmetric version of the above proposition, where $0 <\bar{s} <1$ and the function
$$ g: [\bar{s} ,1]\to \left[ (1-\lambda) \bar{s},1\right]$$
satisfies
$$ g(1) =g' (1) =1 \qquad \mbox{and} \qquad g(\bar{s})/\bar{s} =g'(\bar{s}) =1-\lambda.$$
See Figure~\ref{fig:g2} for an illustration.

\begin{figure}[htbp]
\centering
\includegraphics[width=0.8\textwidth]{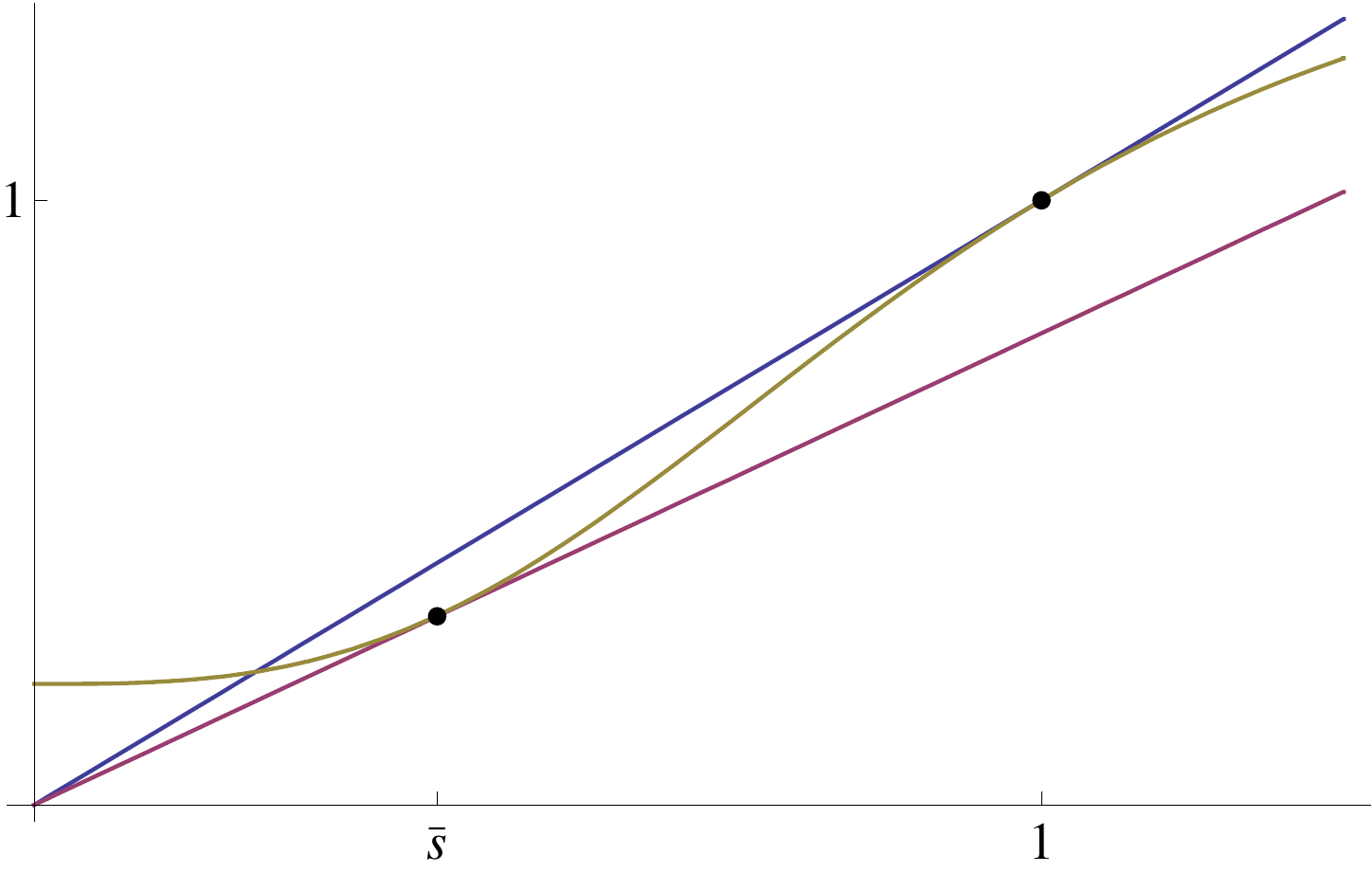}
\caption{Smooth pasting conditions for the function~$g$, for $\theta>1$.}\label{fig:g2}
\end{figure}

Define now
$$m_t =\sup\limits_{0\le u\le t} S_u, \qquad 0\le t\le\sigma_1,$$
as the \emph{running maximum process} of~$S$, where
$$\sigma_1 =\inf \{ t\geq \varrho_0=0: \tfrac{S_t}{m_t} \le \bar{s}\}.$$
Likewise, define
$$M_t=\inf\limits_{\sigma_1\le u\le t} S_u, \qquad \sigma_1 \le t\le \varrho_1,$$
as the \emph{running minimum process} of~$S$, where
$$\varrho_1=\inf\{t\geq \sigma_1 :\tfrac{S_t}{M_t} \geq \tfrac{1}{\bar{s}}\}$$
etc. Continuing in an obvious way, we can again extend~$m$ continuously to~$\mathbb{R}_+$ by setting  
$$m_t=M_t/\bar{s}, \quad \mbox{for}\ t \in \bigcup_{n=0}^\infty \llbracket \sigma_n, \varrho_n \rrbracket.$$
For
$$\tilde{S}_t=m_t g\left(\tfrac{S_t}{m_t}\right), \qquad \T,$$
we then again obtain the conclusion of the above proposition, i.e.,
$$d\tilde{S}_t =g'\left(\tfrac{S_t}{m_t}\right) dS_t +\tfrac{1}{2m_t} g'' \left(\tfrac{S_t}{m_t}\right) d \langle S,S\rangle_t.$$
\end{remark}

\section{Heuristic derivation of the function~$g$}\label{sec:3}
 We now explain on an intuitive level how to come up with a candidate function~$g$ that satisfies the smooth pasting conditions from Section~\ref{sec:2} and leads to a process $\tilde{S}_t=m_t g(S_t/m_t)$, whose $\log$-optimal portfolio keeps the positions in stock and bond constant as long as $S_t/m_t$ lies in the interior of $[1,\bar{s}]$ (resp.\ $[\bar{s},1]$ in the setting of Remark~\ref{rem:symmetric}).

To this end, suppose we start at $S_{t_0}=1=m_{t_0}$ with a portfolio $(\varphi^0_{t_0},\varphi_{t_0})$ such that the \emph{proportion~$\pi$ of total wealth invested into stocks in terms of the ask price~$S$}
\begin{equation}\label{eq:minfrac}
\pi_{t_0}=\frac{\varphi_{t_0} S_{t_0}}{\varphi^0_{t_0}+\varphi_{t_0} S_{t_0}}=\frac{1}{1+\varphi^0_{t_0}/\varphi_{t_0}}
\end{equation} 
lies on the buying side of the no-trade region. 

First suppose that the Merton proportion $\theta=\mu/\sigma^2$ lies in the interval $(0,1)$. This implies that, in the model without transaction costs, the optimal holdings~$\varphi^0$ in bonds and~$\varphi$ in stocks are always strictly positive. We suppose (and shall later prove) that the same holds true under transaction costs. Then if~$S$ starts a positive excursion from level~$S_{t_0}$ at time~$t_0$, the processes $(m_t)_{t \ge t_0}$, $(\varphi^0_t)_{t \ge t_0}$ and $(\varphi_t)_{t \ge t_0}$ remain constant. The fraction of stocks~$\pi$ starts this positive excursion from~$\pi_{t_0}$, too, until~$S$ reaches some level $\bar{s}>1$, where~$\pi$ is positioned at the selling boundary of the no-trade region. At this time $t_1$, the fraction of wealth held in stocks has evolved to 
\begin{equation}\label{eq:maxfrac}
\frac{\varphi_{t_0}\bar{s}}{\varphi^0_{t_0}+\varphi_{t_0}\bar{s}}=\frac{1}{1+\varphi^0_{t_0}/(\varphi_{t_0}\bar{s})}.
\end{equation}
Now suppose that, during this time interval $[t_0,t_1]$, the process~$\tilde{S}$ is given by
$$\tilde{S_t}=g(S_t),$$
for some $C^2$-function~$g$ that we now want to determine. It\^o's formula  and \eqref{eq:stock} yield
$$\frac{dg(S_t)}{g(S_t)}=\Big(\frac{\mu g'(S_t) S_t+ \frac{\sigma^2}{2} g''(S_t)S_t^2}{g(S_t)}\Big)dt+\Big(\frac{\sigma g'(S_t) S_t}{g(S_t)}\Big)dW_t=:\tilde{\mu}_tdt+\tilde{\sigma}_tdW_t.$$
The mean-variance ratio of the process $\tilde{S}=g(S_t)$ is therefore given by
\begin{equation}\label{eq:prop1}
\frac{\tilde{\mu}_t}{\tilde{\sigma}^2_t}=\frac{g(S_t)[\mu g'(S_t)S_t+ \frac{\sigma^2}{2} g''(S_t)S_t^2]}{\sigma^2 g'(S_t)^2 S_t^2}.
\end{equation}
Let us now consider the fraction~$\tilde{\pi}$ of wealth invested in the stock divided by the total wealth at time~$t$, if we evaluate the stock at price~$\tilde{S}$. We obtain
\begin{equation}\label{eq:prop2}
\tilde{\pi}_t=\frac{\varphi_t \tilde{S}_t}{\varphi^0_t+\varphi_t \tilde{S}_t}=\frac{g(S_t)}{c+g(S_t)},
\end{equation}
where~$c$ is defined by
$$c:=\varphi^0_t/\varphi_t=\varphi^0_{t_0}/\varphi_{t_0}, \quad \mbox{for } t \in [t_0,t_1].$$
Note that~$c$ remains constant as long as~$\tilde{S}_t$ lies in the interior of the bid-ask spread $[(1-\lambda)S_t,S_t]$, i.e., for $t \in [t_0,t_1]$. Indeed, the idea is to construct~$\tilde{S}$ in such a way that the frictionless optimizer $(\varphi^0,\varphi)$ associated to~$\tilde{S}$ only moves on the set $\{\tilde{S}_t=(1-\lambda)S_t\}\cup \{\tilde{S}_t=S_t\}$.

Here comes the decisive argument. Merton's rule~\eqref{eq:merton} tells us that the $\log$-optimal portfolio for the (frictionless) process~$\tilde{S}$ must have the following property: The ratio~\eqref{eq:prop2} of wealth invested in the stock $\varphi_t \tilde{S}_t$ divided by the total wealth $\varphi^0_t+\varphi_t \tilde{S}_t$ must be equal to the mean-variance ratio~\eqref{eq:prop1}. A short calculation shows that this equality is tantamount to the following ODE for~$g$:
\begin{equation}\label{eq:ODE}
g''(s)=\frac{2g'(s)^2}{c+g(s)}-\frac{2\mu g'(s)}{\sigma^2 s}, \quad 1 \le s \le \bar{s}.
\end{equation}
We still need the corresponding boundary conditions. Since the proportion of wealth held in stocks started at the buying boundary at time~$t_0$, the shadow price must equal the higher ask price there, i.e.\ $1=S_{t_0}=\tilde{S}_{t_0}=g(S_{t_0})=g(1)$. Likewise, since the proportion of wealth held in stocks has moved to the selling boundary when the ask price $S_t$ reaches level~$\bar{s}$, we must have $g(\bar{s})=(1-\lambda)\bar{s}$ such that~$\tilde{S}_t$ coincides with the lower ask price $(1-\lambda)S_t$. Since the boundary~$\bar{s}$ is not known a priori, we need some additional boundary conditions, which we can heuristically derive as follows. Since we want~$\tilde{S}_t$ to remain in the bid-ask spread $[(1-\lambda)S_t,S_t]$, the ratio $\tilde{S}_t/S_t=g(S_t)/S_t$ must remain within $[1-\lambda,1]$ as~$S_t$ moves through $[1,\bar{s}]$. Therefore, its diffusion coefficient should tend to zero as~$S_t$ approaches either~$1$ or~$\bar{s}$. It\^o's formula yields that the diffusion coefficient of $g(S_t)/S_t$ is given by $S_t^{-2}(g'(S_t)S_t-g(S_t))$. Together with $g(1)=1$ and $g(\bar{s})=(1-\lambda)\bar{s}$, this implies that we should have $g'(1)=1$ and $g'(\bar{s})=(1-\lambda)$. These are precisely the smooth pasting conditions from  Section~\ref{sec:2}.

Imposing the two boundary conditions $g(1)=g'(1)=1$, the general  \emph{closed-form} solution of the ODE~\eqref{eq:ODE} is given by
\begin{equation}\label{eq:explicit}
g(s)=\frac{-cs+(2\theta-1+2c\theta)s^{2\theta}}{s-(2-2\theta+c(2\theta-1))s^{2\theta}},
\end{equation}
unless $\theta = \frac{1}{2}$, which is a special case that can be treated analogously (cf.\ Lemma~\ref{lem:g} below). For given $\lambda>0$, it remains to determine~$\bar{s}$ and~$c$ such that $g(\bar{s})=(1-\lambda)\bar{s}$ and $g'(\bar{s})=(1-\lambda)$. This is equivalent to requiring $g(\bar{s})=(1-\lambda)\bar{s}$ and  $g(\bar{s})=\bar{s} g'(\bar{s})$. Plugging~\eqref{eq:explicit} into the latter condition yields
\begin{equation}\label{eq:heuristics}
\bar{s}=\bar{s}(c)=\left(\frac{c}{(2\theta-1+2c\theta)(2-2\theta-c(2\theta-1))}\right)^{1/(2\theta-1)}.
\end{equation}
To determine $c$ from the Merton proportion $\theta=\mu/\sigma^2$ and the transaction costs $\lambda$, insert~\eqref{eq:explicit} and~\eqref{eq:heuristics} into the remaining condition $g(\bar{s})=(1-\lambda)\bar{s}$. We find that~$c$ must solve
\begin{equation}\label{eq:heuristicc}
 \left(\frac{c}{(2\theta-1+2c\theta)(2-2\theta-c(2\theta-1)}\right)^{\frac{1-\theta}{\theta-1/2}}-\frac{1}{1-\lambda}(2\theta-1+2c\theta)^2=0.
 \end{equation}
Once we have determined~$c$, this yields~$\bar{s}$ and, via~\eqref{eq:minfrac} and~\eqref{eq:maxfrac}, the lower resp.\ upper limits $1/(1+c)$ and 
$1/(1+c/\bar{s})$ for the fraction~$\pi$ of total wealth held in stocks in terms of the ask price~$S$. It does not seem to be possible to determine~$c$ in closed form from~\eqref{eq:heuristicc} as a function of~$\lambda$. However, the above representation easily leads to fractional Taylor expansions in terms of $\lambda>0$ for~$c$, $\bar{s}$ and the lower resp.\ upper limits $1/(1+c)$ and $1/(1+c/\bar{s})$ for~$\pi_t$. This completes the heuristics for the case $\theta \in (0,1) \backslash \{\frac{1}{2}\}$. As mentioned above, the case $\theta=\frac{1}{2}$ can be dealt with in an analogous way except for a different solution of the ODE for~$g$ (see Lemma~\ref{lem:g} below). 

Now consider a Merton proportion $\theta=\mu/\sigma^2 \in (1,\infty)$. In this case, the $\log$-investor in the price process~$S$ without transaction costs goes short in the bond, i.e., chooses $\varphi^0<0$ and $\varphi>0$. We again suppose (and subsequently verify in Section~\ref{sec:6}) that this remains true in the presence of transaction costs. Then if~$S$ starts a \emph{negative} excursion from level~$S_{t_0}$ at time~$t_0$, the processes $(m_t)_{t \ge t_0}$, $(\varphi^0_t)_{t \ge t_0}$ and $(\varphi_t)_{t \ge t_0}$ remain constant. The fraction~$\pi$ of stocks in turn starts a \emph{positive} excursion from~$\pi_{t_0}$, until~$S$ reaches some level $\bar{s}<1$, where~$\pi$ is positioned at the higher selling boundary of the no-trade region (see Figure~\ref{fig:g2}). The remaining arguments from above can now be carried through accordingly, by replacing $[1,\bar{s}]$ with $[\bar{s},1]$. Consequently, one ends up precisely in the setup of Remark~\ref{rem:symmetric}.

Finally, consider the degenerate case $\theta=\mu/\sigma^2=1$. Then the ODE~\eqref{eq:ODE} for~$g$ complemented with the boundary conditions $g(1)=g'(1)=1$ already implies $g(s)=s$ and $c=0$. Since the other boundary condition $g(\bar{s})=(1-\lambda)\bar{s}$ and $g'(\bar{s})=1-\lambda$ cannot hold in this case (except for $\lambda=0$), we formally have $\bar{s}=\infty$. This means that the shadow price~$\tilde{S}$ coincides with the ask price~$S$ and the corresponding optimal fraction of wealth held in stock evaluated at price $\tilde{S}=S$ is constantly equal to one, since the lower and upper boundaries $1/(1+c)$ and $1/(1+c/\bar{s})$ both become~$1$ for $c = 0$ and $\bar{s} = \infty$. This is also evident from an economic point of view, since Merton's rule~\eqref{eq:merton} implies that the optimal strategy for $\tilde{S}=S$ without transaction costs consists of refraining from any trading after converting the entire initial endowment into stocks at time zero.

\section{Existence of the candidates}\label{sec:4}

To show that the heuristics from the previous section indeed lead to well-defined objects, we begin with the following elementary observations. Their straightforward but tedious proofs are deferred to Appendix~\ref{appendix}.

\begin{lemma}\label{l:existence}
Fix $0 < \theta \neq 1$, and let
\begin{align*}
&f(c)=\\
 &\begin{cases}
 \left(\frac{c}{(2\theta-1+2c\theta)(2-2\theta-c(2\theta-1)}\right)^{\frac{1-\theta}{\theta-1/2}}-\frac{1}{1-\lambda}(2\theta-1+2c\theta)^2, &\mbox{if  $\theta \in (0,\infty) \backslash \{\frac{1}{2},1\}$},\\
\exp\left(\frac{c^2-1}{c}\right)-\frac{1}{1-\lambda}c^2, &\mbox{if $\theta=\frac{1}{2}$}.
\end{cases}
\end{align*}
 Then there exists a unique solution to $f(c)=0$ on $(\frac{1-\theta}{\theta},\infty)$ if $\theta \in (0,\frac{1}{2}]$, on $(\frac{1-\theta}{\theta},\frac{1-\theta}{\theta-1/2})$ if $\theta \in (\frac{1}{2},1)$, resp.\ on $(\frac{1-\theta}{\theta},0)$ if $\theta>1$.
\end{lemma}

For fixed $0<\theta \neq 1$ and~$c$ as in Lemma~\ref{l:existence}, we can now define the real number~$\bar{s}$ as motivated in the heuristics for $\theta \neq \frac{1}{2}$.

\begin{lemma}\label{lem:s}
Fix $0 < \theta \neq 1$. Then for~$c$ as in Lemma~\ref{l:existence},
\begin{equation}\label{eq:s}
\bar{s}=\begin{cases}  \left(\dfrac{c}{(2\theta-1+2c\theta)(2-2\theta-c(2\theta-1)}\right)^{1/(2\theta-1)} &\mbox{if $\theta \in (0,\infty) \backslash \{\frac{1}{2},1\}$},\\ 
\exp\left(\dfrac{c^2-1}{c}\right)
\vphantom{X^{X^{X^{X^{X^{X^{X}}}}}}}, 
&\mbox{if $\theta=\frac{1}{2}$} \end{cases}
\end{equation}
is well-defined and lies in $(1,\infty)$ if $\theta \in (0,1)$, resp.\ in $(0,1)$ if $\theta \in (1,\infty)$. Moreover, we have $c/\bar{s} \in (0,\infty)$ if $\theta \in (0,1)$ resp.\ $c/\bar{s} \in (-1,0)$ if $\theta>1$. 
\end{lemma}

Now we can verify by insertion that the candidate function~$g$ has the properties derived in the heuristics above.

\begin{lemma}\label{lem:g}
For $0 < \theta \neq 1$ as well as $c$ and $\bar{s}$ as in Lemmas~\ref{l:existence} resp.~\ref{lem:s}, define
\begin{equation}\label{eq:def g}
g(s):=\begin{cases}  \dfrac{-c s + ({2\theta}-1+2c \theta) s^{2\theta}}{s -(2-2\theta-c(2\theta-1)) s^{2\theta}} &\mbox{if $\theta \in (0,\infty) \backslash \{\frac{1}{2},1\}$},\\  \dfrac{\vphantom{X^{X^{X^X}}}(c+1)+c \log(s)}{c+1-\log(s)} &\mbox{if   $\theta=\frac{1}{2}$,}\end{cases}
\end{equation}
on $[1,\bar{s}]$ if $\theta \in (0,1)$, resp.\ on $[\bar{s},1]$ if $\theta \in (1,\infty)$. Then $g'>0$. Moreover, $g$ takes values in $[1,(1-\lambda)\bar{s}]$ (for $\theta \in (0,1)$) resp.\  $[(1-\lambda)\bar{s},1]$ (for $\theta \in (1,\infty)$), solves the ODE
\begin{equation}\label{eq:odeg}
g''(s)=\frac{2g'(s)^2}{c+g(s)}-\frac{2\theta g'(s)}{s},
\end{equation}
and satisfies the boundary conditions
$$g(1)=g'(1)=1, \quad g(\bar{s})=(1-\lambda)\bar{s}, \quad g'(\bar{s})=1-\lambda.$$
\end{lemma}

\section{The shadow price process and its log-optimal portfolio}\label{sec:6}

With Proposition~\ref{proG5} and the function $g$ from Lemma~\ref{lem:g} at hand, we can now construct a shadow price $\tilde{S}$ and determine its $\log$-optimal portfolio. To this end, let $g$ be the function from Lemma~\ref{lem:g}. Then, for the process $m$ as defined in Section~\ref{sec:2}, Proposition~\ref{proG5} yields that 
$$\tilde{S}_t:=m_t g\left(\tfrac{S_t}{m_t}\right)$$
is an It\^o process satisfying the stochastic differential equation
\begin{equation}\label{eq:sde}
d\tilde{S}_t/\tilde{S_t}=\tilde{\mu}\left(\tfrac{S_t}{m_t}\right)dt+\tilde{\sigma}\left(\tfrac{S_t}{m_t}\right)dW_t, \quad \tilde{S}_0=1,
\end{equation}
with drift and diffusion coefficients  
\begin{align*}
\tilde{\mu}(s) = \frac{\sigma^2 g'(s)^2 s^2}{g(s)(c+g(s))}, \qquad \tilde{\sigma}(s) =\frac{\sigma g'(s) s}{g(s)} .
\end{align*}
Note that we have replaced $g''$ in Proposition~\ref{proG5} with the expression provided by the ODE~\eqref{eq:odeg} from Lemma~\ref{lem:g}. Also notice that $\tilde{\mu}$ and $\tilde{\sigma}$ are continuous and hence bounded on $[1,\bar{s}]$,
and~$\tilde{\sigma}$ is also bounded away from zero.

For It\^o processes with bounded drift and diffusion coefficients, the solution to the $\log$-optimal portfolio problem is well-known (cf.\ e.g.~\cite[Example 6.4]{karatzas.al.91}). This leads to the following result.

\begin{theorem}\label{thm:shadow}
Fix $0 < \theta \neq 1$ and let the stopping times $(\varrho_n)_{0 \le n \le \infty}$, $(\sigma_n)_{1 \le n \le \infty}$ and the process $m$ be defined as in Section~\ref{sec:2}. For the function $g$ from Lemma~\ref{lem:g}, set $\tilde{S}_t=m_t g(\frac{S_t}{m_t}).$

Then the $\log$-optimal portfolio $(\varphi^0,\varphi)$ in the frictionless market with price process~$\tilde{S}$ exists and is given by $(\varphi^0_{0-},\varphi_{0-})=(x,0)$, $(\varphi^0_0,\varphi_0)=(\frac{c}{c+1}x,\frac{1}{c+1}x)$ and
$$ 
\varphi^0_t=\begin{cases} \varphi^0_{\varrho_{k-1}}\left(\frac{m_t}{m_{\varrho_{k-1}}}\right)^{\frac{1}{c+1}} &\mbox{on} \  \bigcup_{k=1}^\infty \llbracket \varrho_{k-1},\sigma_k \rrbracket, \\ \varphi^0_{\sigma_k} \left(\frac{m_t}{m_{\sigma_{k}}}\right)^{\frac{(1-\lambda)\bar{s}}{c+(1-\lambda)\bar{s}}} &\mbox{on} \  \bigcup_{k=1}^\infty \llbracket \sigma_{k},\varrho_k \rrbracket, \end{cases}
$$
as well as
$$ 
\varphi_t=\begin{cases} \varphi_{\varrho_{k-1}}\left(\frac{m_t}{m_{\varrho_{k-1}}}\right)^{-\frac{c}{c+1}} &\mbox{on} \ \bigcup_{k=1}^\infty \llbracket \varrho_{k-1},\sigma_k \rrbracket, \\ \varphi_{\sigma_k}\left(\frac{m_t}{m_{\sigma_{k}}}\right)^{-\frac{c}{c+(1-\lambda)\bar{s}}} &\mbox{on} \ \bigcup_{k=1}^\infty \llbracket \sigma_{k},\varrho_k \rrbracket. \end{cases}
$$
 The corresponding optimal fraction of wealth invested into stocks is given by
 \begin{equation}\label{eq:tildepi}
 \tilde{\pi}_t=\frac{\varphi_t \tilde{S}_t}{\varphi^0_t+\varphi_t \tilde{S}_t}=\frac{1}{1+c/g(\frac{S_t}{m_t})}.
 \end{equation}
 \end{theorem}
 
\begin{proof}
By~\eqref{eq:sde}, $\tilde{S}$ is an It\^o process with bounded coefficients. Since, moreover, $\tilde{\mu}/\tilde{\sigma}^2$ is also bounded, Merton's rule as in~\cite[Example 6.4]{karatzas.al.91} implies that the optimal proportion of wealth invested into stocks is given by
$$\frac{\tilde{\mu}(\frac{S_t}{m_t})}{\tilde{\sigma}^2(\frac{S_t}{m_t})}=\frac{1}{1+c/g(\frac{S_t}{m_t})}.$$
On the other hand, the adapted process $(\varphi_t^0,\varphi_t)_{ t \ge 0}$ is continuous and hence predictable. By definition, 
\begin{equation}\label{eq:rec}
\varphi^0_t=c m_t \varphi_t, \quad t \ge 0.
\end{equation}
For any $k \in \mathbb{N}$, It\^o's formula and~\eqref{eq:rec} now yield
$$d\varphi^0_t+\tilde{S}_td\varphi_t= \left[\left(\frac{m_t}{m_{\varrho_{k-1}}}\right)^{-c/(c+1)}\frac{1}{c+1}\left(\frac{\varphi^0_{\varrho_{k-1}}}{m_{\varrho_{k-1}}}-c\varphi_{\varrho_{k-1}}\right)\right]dm_t=0,$$
on $\llbracket \rho_{k-1},\sigma_k \rrbracket$ and likewise on $\llbracket \sigma_k, \rho_k \rrbracket$. Therefore $(\varphi^0,\varphi)$ is self-financing. Again by~\eqref{eq:rec}, the fraction
$$\frac{\varphi_t \tilde{S_t}}{\varphi^0_t+\varphi_t \tilde{S}_t}=\frac{1}{1+c/g(\frac{S_t}{m_t})}$$
of wealth invested into stocks when following $(\varphi^0,\varphi)$ coincides with the Merton proportion computed above. Hence $(\varphi^0,\varphi)$ is $\log$-optimal and we are done.
\end{proof}

 In view of \eqref{eq:tildepi} and Lemma \ref{lem:g}, it is optimal in the frictionless market with price process $\tilde{S}$ to keep the fraction $\tilde{\pi}$ of wealth in terms of $\tilde{S}$ invested into stocks in the interval $[(1+c)^{-1},(1+c/((1-\lambda)\bar{s}))^{-1}]$. By definition of $\varphi^0$ and $\varphi$, no transactions take place while $\tilde{\pi}$ moves in the interior of this \emph{no-trade region in terms of $\tilde{S}$}. 
 
 As was kindly pointed out to us by Paolo Guasoni, the no-trade region in terms of $\tilde{S}$ is \emph{symmetric relative to the Merton proportion $\theta$}. Indeed, after inserting $(1-\lambda)\bar{s}=g(\bar{s})$, \eqref{eq:def g}, and~\eqref{eq:s}, rearranging yields that $(1+c)^{-1}+(1+c/((1-\lambda)\bar{s}))^{-1}=2\theta$. Hence
 \begin{equation}\label{eq:symmetric}
 \theta-\frac{1}{1+c}=\frac{1}{1+c/((1-\lambda)\bar{s})}-\theta.
 \end{equation}

{}From Theorem~\ref{thm:shadow} we can now obtain that $\tilde{S}$ is a shadow price.

\begin{corollary}\label{cor:shadow}
For $0 < \theta \neq 1$, the process $\tilde{S}$ from Theorem~\ref{thm:shadow} is a shadow price in the sense of Definition~\ref{def:shadow} for the bid-ask process $((1-\lambda)S,S)$. For $\theta=1$, the same holds true by simply setting $\tilde{S}=S$.
\end{corollary}

\begin{proof}
First consider the case $0 < \theta \neq 1$. In view of Proposition~\ref{proG5}, $\tilde{S}=m g(S/m)$ takes values in the bid-ask spread $[(1-\lambda)S,S]$. By Theorem~\ref{thm:shadow}, the $\log$-optimal portfolio $(\varphi^0,\varphi)$ for $\tilde{S}$ exists. Moreover, since $m$ only increases (resp.\ decreases) on $\{S_t=\bar{s}m_t\}$ (resp.\ $\{S_t=m_t\}$), the number of stocks $\varphi$ only increases (resp.\ decreases) on $\{S_t=m_t\}=\{ \tilde{S}_t=S_t\}$ (resp.\ $\{S_t=\bar{s}m_t\}=\{\tilde{S}_t=(1-\lambda)S_t\}$) by definition of $\varphi$. This shows that~$\tilde{S}$ is a shadow price. 

For $\theta=1$, it follows from~\cite[Example 6.4]{karatzas.al.91} that the optimal strategy for the frictionless market $\tilde{S}=S$ transfers all wealth into stocks at time $t=0$ and never trades afterwards, i.e.\ $\varphi^0_t=0$ and $\varphi_t=x$ for all $t \ge 0$. Hence it is of finite variation, the number of stocks never decreases and only increases at time $t=0$ where $\tilde{S}_0=S_0$. This completes the proof. 
\end{proof}

If $\theta \in (0,1)$, Corollary~\ref{cor:tac} combined with Corollary~\ref{cor:shadow} shows that $(\varphi^0,\varphi)$ is also growth-optimal for the bid-ask process $((1-\lambda)S,S)$. The corresponding fraction of wealth invested into stocks in terms of the ask price $S$ is given by 
$$\pi=\frac{\varphi S}{\varphi^0+\varphi S}=\frac{1}{1+\varphi^0/(\varphi S)}=\frac{1}{1+\frac{m}{S}c},$$
where we have used $\varphi^0=cm\varphi$ for the last equality. Hence, the fraction $\pi$ is always kept in the \emph{no-trade-region}  $[(1+c)^{-1},(1+c/\bar{s})^{-1}]$ \emph{in term of $S$}. Note that this interval always lies in $(0,1)$, since $c>0$ and $\bar{s}>1$ by Lemmas~\ref{l:existence} and~\ref{lem:s}.

For $\theta=1$, the investor always keeps his entire wealth invested into stocks.

If $\theta \in (1,\infty)$, one cannot apply Corollary~\ref{cor:tac} directly. However, in the present setting a corresponding statement still holds, provided that the transaction costs $\lambda$ are \emph{sufficiently small}.

\begin{lemma}\label{lem:lambdasmall}
Fix $\theta \in (1,\infty)$ and let $(\varphi^0,\varphi)$ be the $\log$-optimal portfolio for the frictionless market with price process $\tilde{S}$ from Theorem~\ref{thm:shadow}. Then there exists $\lambda_0>0$ such that, for all $0<\lambda<\lambda_0$, the portfolio $(\varphi^0,\varphi)$ is also growth-optimal for the bid-ask process $((1-\lambda)S,S)$.
\end{lemma}

\begin{proof}
First note that $c \in (-1,0)$ by Lemma~\ref{lem:s}. Moreover, the function $g$ is increasing and maps $[\bar{s},1]$ to $[(1-\lambda)\bar{s},1]$ by Lemma~\ref{lem:g}. By~\eqref{eq:tildepi}, the fraction $\tilde{\pi}$ of wealth in terms of $\tilde{S}$ invested into stocks therefore takes values in the interval $[(1+c)^{-1},(1+c/((1-\lambda)\bar{s}))^{-1}]$. Together with $\varphi \ge 0$, this yields the estimate
\begin{equation}\label{eq:lowerbound}
\begin{split}
V_T(\varphi^0,\varphi) &\ge \tilde{V}_T(\varphi^0,\varphi)-\lambda \tilde{\pi}_T \tilde{V}_T(\varphi^0,\varphi)\\
                                         &\ge \left(1-\frac{\lambda}{1+c/((1-\lambda)\bar{s})}\right) \tilde{V}_T(\varphi^0,\varphi).
\end{split}                                         
\end{equation}                                         
By Proposition~\ref{prop:s c expans} below, there exists $\lambda_0>0$ such that
$$\left(1-\frac{\lambda}{1+c/((1-\lambda)\bar{s})}\right)>0 \quad \mbox{ for all}\ 0<\lambda<\lambda_0.$$
Since $\tilde{V}(\varphi^0,\varphi)$ is nonnegative, this shows that $(\varphi^0,\varphi)$ is admissible for the bid-ask process $((1-\lambda)S,S)$, if $0<\lambda<\lambda_0$. The remainder of the assertion now follows as in the proof of Corollary~\ref{cor:tac} by combining the above estimate~\eqref{eq:lowerbound} with the obvious upper bound $V_T(\varphi^0,\varphi) \le\tilde{V}_T(\varphi^0,\varphi)$.
\end{proof}

If Lemma \ref{lem:lambdasmall} is in force, the growth-optimal portfolio under transaction costs for $\theta>1$ also keeps the fraction $\pi$ of stocks in terms of the ask price $S$ in the interval $[1/(1+c), 1/(1+c/\bar{s})]$. In particular, since $c \in (-1,0)$ and $c/\bar{s} \in (-1,0)$, this now entails always going short in the bond, i.e.\ both boundaries of the no-trade region lie in the interval $(1,\infty)$.

\subsection{The optimal growth rate}
We now want to compute the optimal growth rate
\begin{equation}\label{A1}
\delta=\limsup_{T\to\i} \frac{1}{T} \E \left[\log(\tilde{V}_T(\varphi^0,\varphi))\right]=\limsup_{T\to\i} \frac{1}{T} \E \left[\int_0^T \frac{\tilde{\mu}^2(\frac{S_t}{m_t})}{2\tilde{\sigma}^2(\frac{S_t}{m_t})}dt\right],
\end{equation} 
where $(\varphi^0,\varphi)$ denotes the $\log$-optimal portfolio for the shadow price $\tilde{S}$ from Theorem~\ref{thm:shadow} and the second equality follows from~\cite[Example 6.4]{karatzas.al.91}. In view of Corollary~\ref{cor:tac} resp.\ Lemma \ref{lem:lambdasmall}, the above constant $\delta$ coincides with the optimal growth rate for the bid-ask process $((1-\lambda)S,S)$.

It follows from the construction in Section~\ref{sec:2} above that the process $S/m$ is a doubly reflected geometric Brownian motion with drift on the interval $[1,\bar{s}]$ (resp.\ on $[\bar{s},1]$ for the case $\theta>1$). Therefore, an ergodic theorem for positively recurrent one-dimensional diffusions (cf., e.g.,~\cite[Sections II.36 and II.37]{borodin.salminen.02}) and elementary integration yield the following result.

\begin{proposition}
Suppose the conditions of Theorem~\ref{thm:shadow} hold. Then the process $S/m$ has stationary distribution
$$\nu(ds)= \begin{cases} {\displaystyle\frac{2\theta-1}{\bar{s}^{2\theta-1}-1}   s^{2\theta-2} \I_{[1,\bar{s}]}(s)ds}, &\mbox{for $\theta \in (0,1) \backslash \{\frac{1}{2}\}$},\\ 
\vphantom{X^{X^{X^{X^{X^{X^{X}}}}}}}
{\displaystyle\frac{1}{\log(\bar{s})}s^{-1} \I_{[1,\bar{s}]}(s)ds}, &\mbox{for $\theta=\frac{1}{2}$},\\
\vphantom{X^{X^{X^{X^{X^{X^{X}}}}}}}
{\displaystyle\frac{2\theta-1}{1-\bar{s}^{2\theta-1}}   s^{2\theta-2} \I_{[\bar{s},1]}(s)ds}, &\mbox{for $\theta \in (1,\infty)$}.\end{cases}$$
Moreover, the optimal growth rate for the frictionless market with price process $\tilde{S}$ and for the market with bid-ask process $((1-\lambda)S,S)$ is given by
\begin{align}
\delta &= \int_1^{\bar{s}} \frac{\tilde{\mu}^2(s)}{2\tilde{\sigma}^2(s)} \nu(ds) \notag \\
&=
\begin{cases}
  {\displaystyle\frac{(2\theta-1)\sigma^2\bar{s}}{2(1+c)(\bar{s} + (-2 - c + 2 \theta (1 + c))\bar{s}^{2\theta})}}
    & \mbox{for $\theta \in (0,\infty) \backslash \{\frac{1}{2},1\}$}, \\
  {\displaystyle\frac{\vphantom{X^{X^X}}\sigma^2}{2(1+c)(1+c - \log \bar{s})}} & \mbox{for $\theta=\frac{1}{2}$},
\end{cases} \label{eq:OG closed form}
\end{align}
where $c$ and $\bar{s}$ denote the constants from Lemmas~\ref{l:existence} resp.~\ref{lem:s}.
\end{proposition}

\begin{remark}
In the degenerate case $\theta=1$, the optimal portfolio $(\varphi^0,\varphi)=(0,x)$ leads to $\tilde{V}(\varphi^0,\varphi)=x S$. Hence $\mathbb{E}[\log(\tilde{V}_T(\varphi^0,\varphi))]=\log(x)+\frac{\sigma^2}{2}T$ and the optimal growth rate is given by $\delta=\sigma^2/2$ as in the frictionless case.
\end{remark}

\section{Asymptotic expansions}\label{sec:7}

Similarly to Jane{\v{c}}ek and Shreve~\cite{janecek.shreve.04} for the infinite-horizon optimal consumption problem, we now determine asymptotic expansions of the boundaries of the no-trade region and the long-run optimal growth rate.

\subsection{The no-trade region}

We begin with the following preparatory result.

\begin{proposition}\label{prop:s c expans}
  For fixed $0<\theta\neq1$ and sufficiently small $\lambda>0$, the functions $c(\lambda)$   
  and~$\bar{s}(\lambda)$ have fractional Taylor expansions of the form
  \begin{align}
    \bar{s} &= 1 + \sum_{k=1}^\infty p_k(\theta)
      \left( \frac{6}{\theta(1-\theta)} \right)^{k/3} \lambda^{k/3}, \label{eq:s expans} \\
    c &= \bar{c} + \sum_{k=1}^\infty q_k(\theta)
      \left( \frac{6}{\theta(1-\theta)} \right)^{k/3} \lambda^{k/3}, \label{eq:c expans}
  \end{align}
  where $\bar{c}=\frac{1-\theta}{\theta}$, and~$p_k$ and~$q_k$ are rational functions
  that can be algorithmically computed.
  (For $\theta>1$, the quantity $1/(1-\theta)^{k/3}$ has to be read as $(-1)^k/(\theta-1)^{k/3}$.)
  The first terms of these expansions are
  \begin{align}
    \bar{s} &= 1 + \left( \frac{6}{\theta(1-\theta)} \right)^{1/3} \lambda^{1/3} +\frac12 \left( \frac{6}{\theta(1-\theta)} \right)^{2/3} \lambda^{2/3} \label{eq:s first terms} \\
    &\qquad + \frac{1}{60}(4-\theta)(\theta+3) \frac{6}{\theta(1-\theta)}
    \lambda + O(\lambda^{4/3}), \notag \\
    c &= \bar{c} + \frac{1-\theta}{2\theta}\left( \frac{6}{\theta(1-\theta)} \right)^{1/3} \lambda^{1/3} +\frac{(1-\theta)^2}{4\theta} \left( \frac{6}{\theta(1-\theta)} \right)^{2/3} \lambda^{2/3} \label{eq:c first terms} \\
    & \qquad -\frac{1}{40\theta}(\theta-2)(\theta-1)(3\theta-2) \frac{6}{\theta(1-\theta)}
    \lambda + O(\lambda^{4/3}). \notag
  \end{align}
\end{proposition}
\begin{proof} 
  By~\eqref{eq:s},
  the quantity~$\bar{s}=\bar{s}(\lambda)$ can be written as $F(c(\lambda))$, where~$F(z)$ is
  analytic at $z=\bar{c}$. (We will again suppress the dependence of~$\bar{s}$ and~$c$ on~$\lambda$
  in the notation.) We focus on~$\theta\neq\tfrac12$, since the case $\theta=\tfrac12$ is an easy modification.
  
  Expanding the rational function $c/(\dots)$ in~\eqref{eq:s} around~$c=\bar{c}$, and appealing to the binomial
  theorem (for real exponent), we find that the Taylor coefficients of
  \begin{align}
    \bar{s} &= \left( 1 + \frac{2\theta(2\theta-1)}{\theta-1} (c-\bar{c}) +\dots \right)^{1/(2\theta-1)} \notag \\
    &= 1 + \frac{2\theta}{1-\theta}(c-\bar{c}) + \dots \label{eq:s func of c}
  \end{align}
  are rational functions of~$\theta$. An efficient algorithm for the latter step, i.e., for calculating the coefficients
  of a power series raised to some power, can be found in Gould~\cite{Go74}.
  Now we insert the series~\eqref{eq:s func of c} into the equation $g(\bar{s})=(1-\lambda)\bar{s}$, i.e., into
  \begin{equation}\label{eq:modif}
    \lambda \bar{s} = \bar{s} - g(\bar{s}).
  \end{equation}
  Performing the calculations (binomial series again), we find that the expansion of the right-hand side
  of~\eqref{eq:modif} around~$\bar{c}$ starts with the third power of $c-\bar{c}$:
  \[
    \bar{s} - g(\bar{s}) = \tfrac{4\theta^4}{3(1-\theta)^2} (c-\bar{c})^3 (1 + O(c-\bar{c})).
  \]
  Dividing~\eqref{eq:modif} by the series~$\tfrac{4\theta^4}{3(1-\theta)^2}(1 + O(c-\bar{c}))$
  (whose coefficients are again computable)
  therefore yields an equation of the form
  \begin{equation}\label{eq:c^3}
    \lambda(a_0 + a_1 (c-\bar{c}) + \dots) = (c-\bar{c})^3.
  \end{equation}
  The series on the left-hand side is an analytic function $a_0 + a_1 z + \dots$ evaluated at $z=c-\bar{c}$, with real Taylor   
  coefficients $a_0,a_1,\dots$ Its coefficients~$a_k$ are computable rational functions of~$\theta$. Moreover,
  the first coefficient $a_0=\tfrac34 (1-\theta)^2/\theta^4$ is non-zero. Hence we can raise~\eqref{eq:c^3} to
  the power~$\tfrac13$ to obtain
  \[
    \lambda^{1/3}a_0^{1/3}(1 + \tfrac{a_1}{3a_0} (c-\bar{c}) + \dots ) = c-\bar{c},
  \]
  where the power series represents again an analytic function. By the Lagrange inversion theorem (see~\cite{deBr81}
  or~\cite[\S~4.7]{Kn98}),
  $c$ is an analytic function of~$\lambda^{1/3}$:
  \begin{equation}\label{eq:expans c lambda}
    c - \bar{c} = a_0^{1/3} \lambda^{1/3} + \tfrac13 a_0^{-1/3}a_1 \lambda^{2/3} + \dots
  \end{equation}
  This is the expansion~\eqref{eq:c expans}. To see that the coefficients are of the
  announced form, note that
  Lagrange's inversion formula implies that the coefficients in~\eqref{eq:expans c lambda} are given by
  \[
    [\lambda^{k/3}](c-\bar{c}) = \tfrac{1}{k}[z^{k-1}]a_0^{k/3}(1 + \tfrac{a_1}{3a_0} z + \dots)^k ,\qquad k\geq1,
  \]
  where the operator~$[z^k]$ extracts the $k$-th coefficient of a power series.
  Since the~$a_k$ are rational functions of~$\theta$, and
  \[
    a_0^{k/3} = \tfrac{(1-\theta)^k}{8^{k/3}\theta^k} \left( \tfrac{6}{\theta(1-\theta)} \right)^{k/3},
  \]
  the expansion of~$c$ is indeed of the form stated in the proposition.

  As for~$\bar{s}$,
  inserting~\eqref{eq:expans c lambda} into~\eqref{eq:s func of c} yields~\eqref{eq:s expans}:
  \[
    \bar{s} = 1 + \tfrac{2\theta}{1-\theta} \left(\tfrac34 \tfrac{(1-\theta)^2}{\theta^4} \right)^{1/3} \lambda^{1/3} + \dots
  \]
  See Knuth~\cite[\S~4.7]{Kn98} for an efficient algorithm to perform this composition of power series.
\end{proof}

Once the existence of expansions of~$\bar{s}$ and~$c$ in powers of~$\lambda^{1/3}$ is established, 
one can also compute the coefficients by inserting an ansatz
\begin{equation}\label{eq:ansatz}
  \bar{s} = 1 + \sum_{k=1}^\infty A_k \lambda^{k/3}, \qquad
  c = \bar{c} + \sum_{k=1}^\infty B_k \lambda^{k/3}
\end{equation}
into the equations
\begin{align}
  g(\bar{s}) &= (1-\lambda)\bar{s}, \label{eq:g eq 1} \\
  g'(\bar{s}) &= 1-\lambda, \label{eq:g eq 2}
\end{align}
and then comparing coefficients (preferably with a computer algebra system).
Implementing this seems somewhat easier (but less efficient) than implementing the preceding proof.
To give some details, let us look at the
expression in the first line of~\eqref{eq:def g}. Note that, by the binomial theorem (for real exponent),
the coefficients~$\tilde{A}_k$ of
\[
  \bar{s}^{2\theta} = \left(1 + \sum_{k=1}^\infty A_k \lambda^{k/3}\right)^{2\theta} = 1 + \sum_{k=1}^\infty
  \tilde{A}_k \lambda^{k/3}
\]
can be expressed explicitly in terms of the unknown coefficients~$A_k$. Performing the convolution with
the ansatz for~$c$, and continuing in a straightforward way (multiplying by constant factors, and appealing once more
to the binomial series, this time with exponent~$-1$), we find that
\begin{equation}\label{eq:coeff comp1}
  g(\bar{s}) = 1 + A_1 \lambda^{1/3} + A_2 \lambda^{2/3} +
  (\tfrac13 \theta A_1^3+A_3-\tfrac13 \theta^2 A_1^2 (A_1+3B_1))\lambda + \dots
\end{equation}
Now insert the ansatz~\eqref{eq:ansatz} into the right-hand side $(1-\lambda)\bar{s}$ of the first equation~\eqref{eq:g eq 1}:
\begin{equation}\label{eq:coeff comp2}
  (1-\lambda)\bar{s} = 1 + A_1 \lambda^{1/3} + A_2 \lambda^{2/3}  + (A_3-1)\lambda + \dots
\end{equation}
Comparing coefficients in~\eqref{eq:coeff comp1} and~\eqref{eq:coeff comp2}
yields an infinite set of polynomial equations for the~$A_k$ and~$B_k$.
Proceeding analogously for the equation~\eqref{eq:g eq 2} yields a second set of equations.
It turns out that the whole collection can be solved recursively for the coefficients~$A_k$
and~$B_k$.
Along these lines the coefficients in~\eqref{eq:s first terms} and~\eqref{eq:c first terms} were calculated.

With the expansions of~$\bar{s}$ and~$c$ at hand,
we can now determine the asymptotic size of the no-trade region.

\begin{corollary}\label{cor:asymp}\label{cor:notradeasymp}
  For fixed $0<\theta\neq1$, the lower and upper boundaries of the no-trade region in terms of the ask price $S$ have the expansions
  \[
    \frac{1}{1+c} = \theta - \left(\frac34 \theta^2(1-\theta)^2\right)^{1/3} \lambda^{1/3} + \frac{3}{20}(2\theta^2-2\theta+1)\lambda + O(\lambda^{4/3})
  \]
  and
  \[
    \frac{1}{1+c/\bar{s}} = \theta + \left(\frac34 \theta^2(1-\theta)^2\right)^{1/3} \lambda^{1/3}  -\frac{1}{20}(26 \theta^2 - 26 \theta + 3) \lambda + O(\lambda^{4/3}),
  \]
  respectively. The size of the no-trade region in terms of $S$ satisfies
  \begin{equation*}
    \frac{1}{1+c/\bar{s}} - \frac{1}{1+c} = (6\theta^2(1-\theta)^2)^{1/3}\lambda^{1/3}
      -\frac{1}{10}(4\theta-3)(4\theta-1) \lambda + O(\lambda^{4/3}).
  \end{equation*}
\end{corollary}

Note that there is no $\lambda^{2/3}$-term in these expansions. We also stress again that
more terms can be obtained, if desired, by using the machinery of symbolic computation.

\begin{proof}
  Insert the expansions
  of~$\bar{s}$ and~$c$ found in Proposition~\eqref{prop:s c expans} into $1/(1+c)$ and
  $1/(1+c/\bar{s})$.
  A straightforward calculation, using the binomial series and amenable to computer algebra,
  then yields the desired expansions.
\end{proof}

\subsection{The optimal growth rate}

\begin{proposition}\label{prop:growthasymp}
  Suppose that $0<\theta\neq1$. As $\lambda\to0$, the optimal growth rate has the asymptotics
  \begin{equation}\label{eq:OG asympt}
    \delta = \frac{\mu^2}{2\sigma^2} - \left(\frac{3\sigma^3}{\sqrt{128}}\theta^2(1-\theta)^2\right)^{2/3} \lambda^{2/3} +  O(\lambda^{4/3}).
  \end{equation}
\end{proposition}

Note that the $\lambda^{1/3}$- as well as the $\lambda$-term vanish in the above expansion. Moreover, the preceding result can again be extended to a full expansion of~$\delta$
in powers of~$\lambda^{1/3}$.

\begin{proof}
  This easily follows from the explicit formula~\eqref{eq:OG closed form}, by proceeding as in the proof of Corollary~\ref{cor:asymp}.
\end{proof}

\subsection{Comparison to Jane{\v{c}}ek and Shreve (2004)}
In order to compare our expansions to the asymptotic results of Jane{\v{c}}ek and Shreve~\cite{janecek.shreve.04}, we rewrite them in terms of a bid-ask spread $((1-\check{\lambda})\check{S},(1+\check{\lambda})\check{S}))$. As explained in the first footnote in the introduction, we set
$$\check{S}=\tfrac{2-\lambda}{2}S, \quad \check{\lambda}=\tfrac{\lambda}{2-\lambda}.$$
Therefore $\check{S}$ also follows a Black-Scholes model with drift rate $\mu$ and volatility $\sigma$ and the fraction of wealth invested into stocks in terms of $\check{S}$ is given by
$$\check{\pi}=\frac{\varphi \check{S}}{\varphi^0+\varphi \check{S}}=\frac{1}{1+\frac{m}{S}\frac{2c}{2-\check{\lambda}}},$$
where we have again used $\varphi^0=cm\varphi$ for the last equality. This yields the expansions
  \[
    \frac{1}{1+2c/(2-\check{\lambda})} = \theta - \left(\frac32 \theta^2(1-\theta)^2\right)^{1/3} \check{\lambda}^{1/3} + \frac{1}{10}(3 - 11 \theta + 11 \theta^2)\check{\lambda}
     + O(\check{\lambda}^{4/3})
  \]
  for the lower boundary and 
    \[
    \frac{1}{1+2c/((2-\check{\lambda})\bar{s})} = \theta + \left(\frac32 \theta^2(1-\theta)^2\right)^{1/3} \check{\lambda}^{1/3}  -\frac{3}{10}(1 - 7 \theta + 7 \theta^2) \check{\lambda} + O(\check{\lambda}^{4/3}),
  \]
for the upper boundary of the no-trade-region in terms of $\check{S}$, respectively. The size of the no-trade region satisfies
  \begin{align*}
    &\frac{1}{1+2c/((2-\check{\lambda})\bar{s})} - \frac{1}{1+2c/(2-\check{\lambda})}\\
     &\qquad = (12\theta^2(1-\theta)^2)^{1/3} \check{\lambda}^{1/3}
      -\tfrac{1}{5}(4\theta-1)(4\theta-3) \check{\lambda}      
       + O(\check{\lambda}^{4/3}).
  \end{align*}
Comparing these expansions to the results of~\cite{janecek.shreve.04} for the infinite-horizon consumption problem (specialized to logarithmic utility), we find that the leading $\check{\lambda}^{1/3}$-terms coincide. Hence, the relative difference between the size of the no-trade region with and without intermediate consumption goes to zero as $\check{\lambda} \to 0$.

The higher order terms for the consumption problem are unknown. However,~\cite{janecek.shreve.04} argue heuristically that -- in second order approximation --  the selling intervention point is closer to the Merton proportion~$\theta$ than the buying intervention point. They also point out that, on an intuitive level, this is due to the fact that consumption reduces the position in the bank account. In line with this, the present setup without consumption leads to symmetric second order $\check{\lambda}^{2/3}$-terms, which in fact vanish. However, note that the no-trade regions in terms of the mid-price $\check{S}$ resp.\ the ask price $S$ are only symmetric up to order $O(\lambda^{2/3})$, unlike the perfectly symmetric no-trade region in terms of the shadow price~$\tilde{S}$ (cf.\ Equation \eqref{eq:symmetric} above).

An extension of the present approach to the infinite-horizon consumption problem as well as to power utility is left to future research.

\section*{Appendix.  Proofs of Lemmas~\ref{l:existence} and~\ref{lem:s}}\label{appendix}

\begin{proof}[of Lemma~\ref{l:existence}] First consider the case $\theta \in (\frac{1}{2},1)$. By insertion, we find $f(\frac{1-\theta}{\theta})=-\frac{\lambda}{1-\lambda}<0$. Moreover, $f(c) \uparrow \infty$ for $c \uparrow \frac{1-\theta}{\theta-1/2}$, such that a solution to $f(c)=0$ exists on $(\frac{1-\theta}{\theta},\frac{1-\theta}{\theta-1/2})$ by the intermediate value theorem. We now show that it is unique.

 Differentiation yields $f'(\frac{1-\theta}{\theta})=-\frac{4\theta\lambda}{1-\lambda}<0$ and $f''(\frac{1-\theta}{\theta})=-\frac{8\theta^2 \lambda}{1-\lambda}<0$. On the other hand, we find $f''(c)>0$ on $(\frac{1-\theta}{\theta-1/2}-\ve,\frac{1-\theta}{\theta-1/2})$ for some $\ve>0$.

Now notice that~$f''(c)$ is increasing on  $(\frac{1-\theta}{\theta},\frac{1-\theta}{\theta-1/2})$. Indeed, we have 
\begin{equation}\label{eq:f'''}
  f'''(c)=\frac{(2-2\theta)\bar{s}(c)^{2-2\theta}}{c^3(2\theta-1+2c\theta)^3(2-2\theta-c(2\theta-1))^3}k(c),
\end{equation}
for $\overline{s}$ as in Equation~\eqref{eq:s} and
\begin{align*}
k(c)=&16 c^6 \theta^4-48c^4 \theta^3(1-5\theta+4\theta^2)-24c^3\theta(1-13\theta+52\theta^2-72\theta^3+32\theta^4)\\
&-6c^2(-2+50\theta-296\theta^2+640\theta^3-584\theta^4+192\theta^5)\\
&-24c(1-\theta)^2(-3+28\theta-56\theta^2+32\theta^3)+16(1-\theta)^3(6-17\theta+12\theta^2).
\end{align*}
By tedious calculations or using Cylindrical Algebraic Decomposition~\cite{Co75} (henceforth CAD), it follows that $k(c)>0$ and in turn $f'''(c)>0$ on  $(\frac{1-\theta}{\theta},\frac{1-\theta}{\theta-1/2})$. 

Consequently, $f''(c) \leq 0$ on $[\frac{1-\theta}{\theta},c_0]$ and $f''(c)>0$ on $(c_0,\frac{1-\theta}{\theta-1/2})$ for some $c_0 \in (\frac{1-\theta}{\theta},\frac{1-\theta}{\theta-1/2})$. Combining this with $f'(\frac{1-\theta}{\theta})<0$ and  $f(c) \uparrow \infty$ for $c \uparrow \frac{1-\theta}{\theta-1/2}$, we find that there exists $c_1 \in (c_0, \frac{1-\theta}{\theta-1/2})$ such that $f'(c) \leq 0$ on $[\frac{1-\theta}{\theta},c_1]$ and $f'(c)>0$ on $(c_1,\frac{1-\theta}{\theta-1/2})$.
Since $f(\frac{1-\theta}{\theta})<0$, this implies that any solution to $f(c)=0$ must lie on $(c_1,\frac{1-\theta}{\theta-1/2})$ and hence is unique, because $f$ is strictly increasing there.

Now let $\theta \in (0,\frac{1}{2})$. Then $f$ is continuous on  $ (\frac{1-\theta}{\theta},\infty)$ and $f(\frac{1-\theta}{\theta})=-\frac{\lambda}{1-\lambda}<0$ as above. Moreover, the first term of $f(c)$ grows like $c^{(1-\theta)/(1/2-\theta)}$ for $c \to \infty$, whereas the second one grows like $c^2$. Since $\frac{1-\theta}{1/2-\theta}>2$ for $\theta>0$, this implies that $f(c) \uparrow \infty$ for $c \uparrow \infty$. Hence a solution $c$ to $f(c)=0$ exists on $ (\frac{1-\theta}{\theta},\infty)$ by the intermediate value theorem, and it remains to show that it is unique. 

To see this, first notice that again, either by tedious calculations or using CAD, the function $k(c)$ from above turns out to be strictly positive, this time on $ (\frac{1-\theta}{\theta},\infty)$. The remainder of the assertion now follows as above.

Next, let $\theta=\frac{1}{2}$. In this case, $f(1)=-\frac{\lambda}{1-\lambda}<0$ and $f(c)\uparrow \infty$ for $c \uparrow \infty$. Hence a solution to $f(c)=0$ exists on $(1,\infty)$ by the intermediate value theorem. We now show its uniqueness. Indeed, 
$$f'''(c)=\frac{1-6c+9c^2-6c^3+3c^4+c^6}{c^6} \exp\left(\frac{c^2-1}{c}\right)$$
is strictly positive on $(1,\infty)$. Since $f''(1)= -\frac{2\lambda}{1-\lambda}<0$ and $f''(c) \to 1$ for $c \to \infty$, this implies that there exists $c_0 \in (1,\infty)$ such that $f''(c) \leq 0$ on $[1,c_0]$ and $f''(c)>0$ on $(c_0,\infty)$. Combined with $f'(1)=-\frac{2\lambda}{1-\lambda}<0$ and $f'(c)\to \infty$ for $c \to \infty$, this shows that there exists $c_1 \in (c_0,\infty)$ such that $f'(c)\leq 0$ on $[1,c_1]$ and $f'(c)>0$ on $(c_1,\infty)$. Since $f(1)<0$, any solution to $f(c)=0$ must therefore lie in $(c_1,\infty)$ and is unique, since $f$ is strictly increasing there. 

Finally, consider the case $\theta>1$. Then the third derivative~$f'''(c)$ increases
on $(\tfrac{1-\theta}{\theta},0)$. This time~$k(c)$ is negative (by CAD), but so is the fraction
in front of it in~\eqref{eq:f'''}. We can now reason as above, so that the proof is complete.
\end{proof}

\begin{proof}[of Lemma~\ref{lem:s}]
For $\theta=\frac{1}{2}$, this follows immediately from Lemma~\ref{l:existence}, which yields $c \ge 1$ and
$$\bar{s}=\exp\left(\frac{c^2-1}{c}\right)=\frac{1}{1-\lambda} c^2 >1.$$

For $\theta \in (0,1) \backslash\{\frac{1}{2}\}$, one easily shows by CAD that
\[
 \frac{c}{(2\theta-1+2c\theta)(2-2\theta-c(2\theta-1))} \in
 \begin{cases}
   (0,1) & \text{if}\ \theta \in (0,\tfrac12) \cup (1,\infty), \\
   (1,\infty) & \text{if}\ \theta \in (\tfrac12,1),
 \end{cases}
\]
hence~$\bar{s}$ is well-defined. Moreover, $c/\bar{s}$ is positive
for $\theta\in(0,1)$, since $c>0$ and $\bar{s}>0$.

Finally, let $\theta>1$. Then clearly $c/\bar{s}<0$, and it remains to show
that $c/\bar{s}>-1$. For $c\in(\tfrac{1-\theta}{\theta},0)$, we have 
\[
  0 < -(2-2\theta-c(2\theta-1)) < 2\theta - 1 +2c\theta.
\]
Hence, by~\eqref{eq:s},
  \begin{align*}
   c/\bar{s} &= - (-c)^{1-\tfrac{1}{2\theta-1}}
      \left( -(2\theta - 1 +2c\theta)(2-2\theta-c(2\theta-1)) \right)^{\tfrac{1}{2\theta-1}} \\
    & > - (-c)^{1-\tfrac{1}{2\theta-1}} (2\theta - 1 +2c\theta)^{\tfrac{2}{2\theta-1}} =: -h(c).
  \end{align*}
  We have to show that $h(c)<1$ for $c\in(\tfrac{1-\theta}{\theta},0)$. By a discussion similar
  to the proof of Lemma~\ref{l:existence}, the function~$h$ has a unique maximum at
  \[
    c_0 = -\tfrac{(2\theta-1)(\theta-1)}{2\theta^2} \in(\tfrac{1-\theta}{\theta},0),
  \]
  and the value of~$h$ at~$c_0$ satisfies
  \begin{align*}
    h(c_0)^{\theta-1/2} &= \left( 1 - \tfrac{3}{2\theta} + \tfrac{1}{2\theta^2} \right)^{\theta-1}
      \left( 2-\tfrac{1}{\theta} \right) \\
    &= \exp((\theta-1)\log( 1 - \tfrac{3}{2\theta} + \tfrac{1}{2\theta^2} )) \left( 2-\tfrac{1}{\theta} \right) \\
    &< \exp((\theta-1)\left( -\tfrac{3}{2\theta} \right)) \left( 2-\tfrac{1}{\theta} \right).
  \end{align*}
  The last quantity has a negative derivative w.r.t.~$\theta$, and equals~$1$ for $\theta=1$.
  Hence it is smaller than~$1$ for $\theta>1$, so that $h(c_0)<1$, which completes the proof.
\end{proof}

\bibliographystyle{acm}
\bibliography{functionalshadow}

\end{document}